\documentclass[conference]{IEEEtran}
\newcount\DraftStatus  % 0 suppresses notes to selves in text
\usepackage{color}
\definecolor{darkgreen}{rgb}{0,0.5,0}
\definecolor{purple}{rgb}{1,0,1}
\newcommand{\draftnote}[2]{\ifnum\DraftStatus=1
	\marginpar{
		\tiny\raggedright
		\hbadness=10000
        \def\baselinestretch{0.8}
        \textcolor{#1}{\textsf{\hspace{0pt}#2}}}
     \fi}
\usepackage{svnkw}
\newtheorem{thm}{Theorem}[section]

\newtheorem{lem}[thm]{Lemma}

\newtheorem{definition}{Definition}
\usepackage{graphicx,color}
\usepackage{algorithm}
\usepackage{algorithmic}

\usepackage{amsmath, amssymb}
\usepackage{breqn}

\usepackage{multirow}
\usepackage[numbers,sort&compress]{natbib}
\usepackage{subfig}
\usepackage{url}
\usepackage{hyperref}
\usepackage{breakurl}
\usepackage{booktabs}

\ifCLASSINFOpdf
\else
\fi

\begin{document}
\title{Differentially Private Query Learning: from Data Publishing to Model Publishing}

\author{\IEEEauthorblockN{Tianqing Zhu}
\IEEEauthorblockA{School of Information Techonolgy\\
Deakin University\\
Burwood, Australia\\
Email: t.zhu@deakin.edu.au}
\and
\IEEEauthorblockN{Ping Xiong}
\IEEEauthorblockA{School of Information and Security Engineering\\
Zhongnan University of Economics and Law\\
Wuhan, China\\
pingxiong@znufe.edu.cn}
\and
\IEEEauthorblockN{Gang Li}
\IEEEauthorblockA{School of Information Techonolgy\\
Deakin University\\
Burwood, Australia\\
Email: gang.li@deakin.edu.au}
\and
\IEEEauthorblockN{Wanlei Zhou}
\IEEEauthorblockA{School of Information Techonolgy\\
Deakin University\\
Burwood, Australia\\
Email: wanlei@deakin.edu.au}
\and
\IEEEauthorblockN{Philip S. Yu}
\IEEEauthorblockA{Department of Computer Science\\
University of Illinois at Chicago\\
Chicago, US\\
Email: psyu@uic.edu}
}


%Telephone: (800) 555--1212, Fax: (888) 555--1212}


\maketitle

\begin{abstract}
With the development of Big Data and cloud data sharing,
privacy preserving data publishing becomes one of the most important topics in the past decade.
As one of the most influential privacy definitions, differential privacy
provides a rigorous and provable privacy guarantee for data publishing.
Differentially private interactive publishing achieves good performance
in many applications;
however,
the curator has to release a large number of queries in a batch or a synthetic dataset in the Big Data era.
To provide accurate non-interactive publishing
results in the constraint of differential privacy,
two challenges need to be tackled:
one is how to decrease the correlation between large sets of queries,
while the other is how to predict on fresh queries.
Neither is easy to solve by the traditional differential privacy mechanism.
This paper transfers the data publishing problem to
a machine learning problem, in which queries are considered as training samples and a prediction model will be released
rather than query results or synthetic datasets.
When the model is published, it can be used to answer current submitted queries and predict results for fresh queries from the public.
Compared with the traditional method,
the proposed prediction model enhances
the accuracy of query results for non-interactive publishing.
Experimental results show that
the proposed solution outperforms traditional
differential privacy in terms of \emph{Mean Absolute Value} on a large group of queries.
This also suggests the learning model
can successfully retain the utility of published queries while preserving privacy.

\end{abstract}


% \fi
\IEEEpeerreviewmaketitle

%=================================================================

\section{Introduction}\label{sec-intro}

With advances in Big Data and online services,
Privacy Preserving Data Publishing (PPDP) has attracted substantial attention~\cite{zhu17survey}.
In the past few years,
a number of privacy definitions and mechanisms have been proposed in the literature.
Among them,
the most influential one is the notion of \emph{differential privacy}~\cite{Dwork1866758},
which constitutes a rigorous and provable privacy definition.
Existing work on differential privacy has mainly focused on interactive publishing,
in which the curator releases query answers one by one~\cite{Dwork1873617}.
However,
curators have to releases a large number of queries or a synthetic dataset in may Big Data scenarios such as machine learning or data mining~\cite{Dwork1873617}.
The existing differentially private method fails to provide accurate results when publishing a large number of queries.~\cite{Li2013ICDT}.
This obstacle hinders the implementation of
differential privacy in Big Data applications.

The difficulty in non-interactive data publishing lies on
the high correlation between multiple queries~\cite{Huang2015DB}.
High correlation leads to large volume of noise adding to the query results.
Given a fixed privacy level, noise is determined by the sensitivity.
The sensitivity is defined to capture the
difference in the query results between the addition or removal of a single record in a dataset.
When answering a query set,
a curator has to calibrate the sensitivity of these queries,
but as deleting one record may affect multiple query answers.
In another word,
the queries in the set may correlate to one another.
Correlations between $m$ queries lead to
higher sensitivity (normally $m$ multiplied by the original sensitivity) than independent queries~\cite{Huang2015DB}.
%The noise calibrated by this higher sensitivity will be added to each query answer.
%The accuracy of the results will be dramatically decreased compared to independent queries~\cite{Li2013ICDT}.
The problem is illustrated in the following example.

\subsection{An Example}\label{sec-example}

Table~\ref{Tab-kbatch} shows a frequency dataset $D$ with $n=4$ variables,
and Table~\ref{Tab-kbatchquery} contains all possible \texttt{range} queries
$F=\{f_1,...,f_{10}\}$.
%Suppose a curator plans to publish this set of query for $D$,
%he will first measure the sensitivity of these queries.
As changing any value in $D$ will impact on at most $6$ query results
(column containing $x_{2}$ or $x_{3}$ in Table~\ref{Tab-kbatchquery}) in $F$,
%According to the definition of differential privacy,
the sensitivity of $F$ is $6$,
which is much higher than the sensitivity of a single query.
The noise $Laplace(\frac{6}{\epsilon})$ will be added to every query in $F$.
The sensitivity of $F$ is $O(n^2)$ and
the variance of the noise per answer is $O(n^4/\epsilon^2)$.
When $n$ is large, the utility of the results will be demolished.
Alternatively,
\emph{Laplace} noise can be added to $D$
and the range query results are generated accordingly.
In this situation,
the sensitivity is $1$,
but the privacy budget has to be divided into $10$ pieces and arranged to each query in $F$,
leading even higher noise than the previous method.

This example shows that traditional publishing methods
introduce a large volume of noise and lead to inaccurate results,
while \emph{how to reduce noise} remains a major problem in non-interactive publishing.

\subsection{Challenges and Rationale}

\begin{table}
\centering
  \caption{Dataset $D$\label{Tab-kbatch}}
  \begin{tabular}{ccc}
\toprule
     Grade & Count & Variable  \\
\midrule
     $90-100$ & $12$ & $x_1$ \\
     $80-89$ & $24$ & $x_2$ \\
     $70-79$ & $6$ & $x_3$ \\
     $60-69$ & $7$ & $x_4$ \\
\bottomrule
\end{tabular}
\end{table}

\begin{table}  \centering
  \caption{Range Queries\label{Tab-kbatchquery}}
  \begin{tabular}{c|c|c|c|c|c|c|c}
  \toprule
  Range Query              \\ \midrule
        $f_{1}$ & $x_{1}$ & + & $x_{2}$ & + & $x_{3}$ & + & $x_{4}$  \\ %\hline
        $f_{2}$ & $x_{1}$ & + & $x_{2}$ & + & $x_{3}$ &   &          \\ %\hline
        $f_{3}$ &         & + & $x_{2}$ & + & $x_{3}$ & + & $x_{4}$  \\ %\hline
        $f_{4}$ & $x_{1}$ & + & $x_{2}$ &   &         &   &          \\ %\hline
        $f_{5}$ &         & + & $x_{2}$ & + & $x_{3}$ &   &           \\ %\hline
        $f_{6}$ &         &   &         &   & $x_{3}$ & + & $x_{4}$  \\ %\hline
        $f_{7}$ & $x_{1}$ &   &         &   &         &   &           \\ %\hline
        $f_{8}$ &         &   & $x_{2}$ &   &         &   &          \\ %\hline
        $f_{9}$ &         &   &         &   & $x_{3}$ &   &         \\ %\hline
        $f_{10}$ &        &   &         &   &         &   & $x_{4}$  \\ %\hline
  \bottomrule
\end{tabular}
\end{table}

The above example illustrates that
two challenges need to be tackled in the non-interactive publishing.
\begin{itemize}
\item
How to decrease correlation among queries?
As the correlation among queries will introduce large noise,
and this high volume of noise must be added to every query according to the definition of differential privacy.
We have to decrease the correlation to reduce the introduced noise.

\item
How to deal with unknown fresh queries?
As the curator cannot know what users will ask after the data has been published,
he/she has to consider all possible queries and adds pre-defined noise.
When we meet with the scenarios in Big Data,
it is impossible to list all queries.
Even if the curator is able to list all queries,
this pre-defined noise will dramatically decrease the utility of the publishing results.
\end{itemize}

Several works have been carried out over the last decade to address the first challenge
by a strategy.
Let's use the above example again.
If we only need to answer $f_{7}$ to $f_{10}$,
the sensitivity of the query set is
decreased to $1$ and other query results can be generated by the combination of $f_{7}$ to $f_{10}$.
For example,
the answer of $f_{6}$ can be generated by $f_{9}$ and $f_{10}$.
The solution looks quite simple and effective and
most existing work are following this strategy to decrease the correlation.
Xiao~{et~al.}~\cite{Xiao2011} proposed a wavelet transformation
to decrease the correlation between queries.
Li~{et~al.}~\cite{Li2013ICDT} applied the \emph{Matrix} method
to transform the set of queries into a suitable workload.
Similarly,
Huang~{et~al.}~\cite{Huang2015DB} transformed the query sets into
a set of orthogonal queries.
However,
they can only partly solve the first challenge,
while the second challenge has not been touched by these methods.
For complex queries, such as similarity queries,
it is hard for a curator to figure out all independent queries.
In addition,
when given a fresh new query,
how did the curator generate the result by combining of old queries is still remain unknown.

We observe that these two challenges can be overcome
by transferring the data publishing problem to
a machine learning problem.
We treat the queries as training samples
which are used to generate a prediction \emph{model},
rather than releasing a set of queries or a synthetic dataset.
For the first challenge, correlation between queries,
we will apply limited queries to train the model.
These limited queries have lower correlation than in the original query set.
If we can guarantee the training queries can cover most possible scenarios,
the output model will have higher prediction capability.

For the second challenge,
the model can be used to predict the remaining queries,
including those fresh queries.
Actually,
the model prediction is to generate the combination of training queries.
Consequently,
the quality of the model is determined by two key factors:
the coverage of the training samples
and the prediction capability of the prediction model.
The model can help to answer unlimited number of complex queries.

In the above example,
if we use $f_7$ to $f_{10}$ as the training samples,
the sensitivity will be diminished to $1$ and the added noise will be decreased accordingly.
The prediction model $M$ can be used to answer $f_1$ to $f_6$.
Because the prediction process will not access original dataset $D$,
it will not consume any privacy budget.

%\subsection{Contributions}

The target of this paper is to propose a novel differentially private
publishing method that can
answer all possible queries with acceptable utility.
We make the following contributions:

\begin{itemize}
  \item We propose a novel differentially private data publishing method, \emph{MLDP},
  which successfully transfers the data publishing problem to a machine learning problem,
  solving the current challenges of non-interactive data publishing.
  This method exploits a means of publishing more types of data,
  such as a data model.

  \item We analyze both the privacy and the utility of \emph{MLDP},
  demonstrating that the \emph{MLDP} satisfies $\epsilon$-differential privacy
  and proving the accuracy bound of the \emph{MLDP}.

%  \item This method exploits a means of publishing more types of data,
%    such as a data model.
%    Previous work on differential privacy have focused on releasing query answers or synthetic datasets.
%    As mentioned earlier, both methods have natural weaknesses.
%    Releasing data structures,
%    such as tables, trees, or graphs, can help to solve some specific application problems.

  \item We use extensive experiments on both real and a simulated dataset
  to prove the effectiveness of the proposed \emph{MLDP}.
  After comparing our method with traditional \emph{Laplace} and other prevalent publishing methods,
  we conclude that the \emph{MLDP} demonstrates better performance when answering a large set of queries.
\end{itemize}

\section{Preliminaries} \label{sec-preliminaries}

\subsection{Notation}

%%% ----------------------------------------------------------------------

We consider a finite \emph{data universe} $\mathcal{X}$
with the size $|\mathcal{X}|$.
Let $r$ be a \emph{record}
with $d$ attributes sampled from the universe $\mathcal{X}$,
while a \emph{dataset} $D$ is
an unordered set of $n$ records from domain $\mathcal{X}$.
Two datasets $D$ and $D'$ are \emph{neighboring} datasets if
they differ in only one record.
A \emph{query} $f$ is a function that
maps dataset $D$ to an abstract range $\mathbb{R}$:
$f: D\rightarrow\mathbb{R}$.
A group of queries is denoted as $F=\{f_{1},...,f_{m}\}$,
and $F(D)$ denotes $\{f_{1}(D),...,f_{m}(D)\}$.
We use symbol $m$ to denote the number of queries in $F$.

The maximal difference on the results of query $f$
is defined as the \emph{sensitivity} $s$,
which determines how much perturbation is required for the private-preserving answer.
To achieve the target,
differential privacy provides
a mechanism $\mathcal{M}$,
which is a randomized algorithm that
accesses the database.
The randomized output is denoted by a circumflex over the notation.
For example,
$\widehat{f}(D)$ denotes the randomized answer of querying $f$ on $D$.

\subsection{Differential Privacy}

The target of \emph{differential privacy} is to mask the difference in the answer of query $f$
between the \emph{neighboring datasets}~\cite{Dwork1866758}.
In $\epsilon$-differential privacy,
parameter $\epsilon$ is defined as the \emph{privacy budget}~\cite{Dwork1866758},
which controls the privacy guarantee level of mechanism $\mathcal{M}$.
A smaller $\epsilon$ represents a stronger privacy.
The formal definition of differential privacy
is presented as follows:

\begin{definition}[$\epsilon$-Differential Privacy]\label{Def-DP}
A randomized algorithm $\mathcal{M}$ gives $\epsilon$-differential privacy for any pair of
\emph{neighboring datasets} $D$ and $D'$,
and for every set of outcomes $\Omega$,
$\mathcal{M}$ satisfies:
\begin{equation}
Pr[\mathcal{M}(D) \in \Omega] \leq \exp(\epsilon) \cdot Pr[\mathcal{M}(D') \in \Omega]
\end{equation}
\end{definition}

\emph{Sensitivity} is a parameter determining how much
perturbation is required in the mechanism with a given privacy level.

\begin{definition}[Sensitivity]\label{Def-GS}~\cite{Dwork1866758}
For a query $f:D\rightarrow\mathbb{R}$,
the sensitivity of $f$ is defined as
\begin{equation}
s=\max_{D,D'} ||f(D)-f(D')||_{1}
\end{equation}
\end{definition}

The \emph{Laplace} mechanism adds \emph{Laplace} noise to the true answer.
We use $Laplace(\sigma)$ to represent the noise sampled from the \emph{Laplace} distribution with the
scaling $\sigma$.
The mechanism is defined as follows:

\begin{definition}[\emph{Laplace} mechanism]\label{Def-LA}~\cite{Dwork1866758}
Given a function $f: D \rightarrow \mathbb{R}$ over a dataset $D$,
the Eq.~\ref{eq-lap} provides the $\epsilon$-differential privacy.
\begin{equation}
\widehat{f}(D)=f(D)+Laplace(\frac{s}{\epsilon})
\end{equation}\label{eq-lap}
\end{definition}

\section{The Machine Learning Differentially Private Publishing method} \label{sec-method}

\subsection{Overview}

\begin{figure}[htbp]
\centering
{
\includegraphics[scale=0.52]{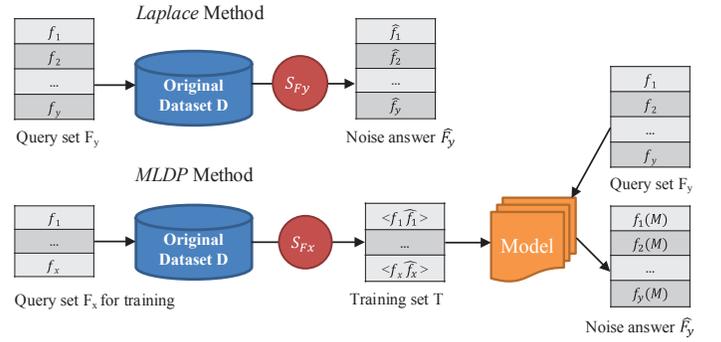}
}
\caption{The overview of the \emph{Laplace }Method and the \emph{MLDP} method}
\label{fig-overview}
\end{figure}

% add a fig to overview
% add utility analysis, current analysis cannot see the advantage.
This section presents the implementation of
the Machine Learning Differentially Private (MLDP) publishing method.
For an original dataset $D$,
suppose a set of query $F_{y}=\{f_{1},...,f_{y}\}$ on the $D$ is waiting to be published.
Fig.~\ref{fig-overview} presents the flow of the
traditional \emph{Laplace} method and the \emph{MLDP} method.
The first flow shows the $Laplace$ method.
When the $F_{y}$ is querying on $D$,
the method will measure the sensitivity of the query set $F_{y}$.
To simplify the notation,
we re-write the sensitivity of a group of queries in definition~\ref{def-SG}.
\begin{definition}[Sensitivity for Correlated Queries]\label{def-SG}
Given a group of queries $F=\{f_1,...,f_{m}\}$ over a dataset $D$,
Equation.~\ref{eq-fsen} provides the sensitivity of $F$.
\begin{equation}\label{eq-fsen}
S_{F}=\max_{D,D'}\sum_{i=1}^{m}||f_{i}(D)-f_{i}(D')||_{1}
\end{equation}
\end{definition}
%Based on the definition,
%the $S_{F}$ in Table~\ref{Tab-kbatchquery} is equal to $6$.
Based on the definition,
the sensitivity of $F_{y}$ is $S_{F_{y}}$.
The \emph{Laplace} noise is calibrated by $Lap(\frac{S_{F_{y}}}{\epsilon})$
and is added to the true answer $F_{y}(D)$.
The \emph{Laplace} method finally outputs the noisy answer $\widehat{F}_{y}$.

The second flow presents the \emph{MLDP} method.
Unlike the traditional \emph{Laplace} method,
it first selects a query set $F_{x}$ for training.
Sensitivity $S_{F_{x}}$ is measured and noise is added to query answers.
Training set $T=<F_{x},\widehat{F_{x}}(D)>$ is carefully selected to make sure $S_{F_{x}}\leq S_{F_{y}}$.
When a learning model is generated,
it will accept the query set $F_{y}$ and make the prediction $F_{y}(M)$.
The \emph{MLDP} method eventually outputs $\widehat{F}_{y}=F_{y}(M)$.

Comparing with the \emph{Laplace} method,
the proposed \emph{MLDP} adds less noise in the training set
than that in the \emph{Laplace} method.
This is because $F_{y}$ normally has correlation with others,
while $F_{x}$ can be selected with lower correlation.
According to Definition~\ref{def-SG},
$S_{F_{x}}$ will be smaller than $S_{F_{y}}$.
A smaller sensitivity leads to less noise.
This helps to solve the first challenge,
high correlation problem in the \emph{Laplace} method.
In addition,
prediction model accepts fresh queries,
which are unknown to the curator before data publishing.
Eventually,
these two properties of \emph{MLDP} help to tackle those two
challenges in traditional \emph{Laplace} method.

\subsection{Implementation of MLDP}

At a high level,
\emph{MLDP} works in three steps:
\begin{itemize}
  \item \emph{Generating training samples}: The curator selects a query set $F_{x}$ with $m$ queries to
  generate training samples.

  \item \emph{Training the model}: The training set is used to train a prediction model $M$.
  In theory, we can select any machine learning algorithm.
  As most of the query answers are numerical values,
  regression algorithms will be more suitable.

  \item \emph{Making prediction}: The model is applied to make prediction of fresh queries $F_{y}$.
  The prediction results $F_{y}(M)$ will be output as the noisy answer of those queries.
\end{itemize}

Algorithm~\ref{Alg-MLDP} illustrates the detail of the \emph{MLDP} method.
At the first step,
%the query set $F$ is divided into $F_{1}$ and $F_{2}$.
%$F_{1}$ will be applied to generate training set,
%so the second step measures the sensitivity of $F_{x}$.
%Because $F_{x}$ is the subset of $F$, $S_{Fx}$ will be smaller than the $S_{F}$ and the noise added to $F_{1}$
%will be diminished accordingly.
we measure the sensitivity of the training sample $F_{x}$.
Because $F_{x}$ is a subset of $F$, $S_{Fx}$ will be smaller than $S_{F}$ and the noise added to $F_{x}$
will be diminished accordingly.
At Step $2$,
\emph{Laplace} noise is added to the true answer $F_{x}(D)$ and we obtain the noisy answer $\widehat{F_{x}}(D)$.
Step $3$ generates the \emph{training set}  $T=<F_{x}, \widehat{F_{x}}(D)>$.
Step $4$ uses $T$ to learn a regression model $M$ and consider it to be a synopsis of the original dataset $D$.
At Step $5$,
model $M$ will be released to the public and
every time public users try to query the dataset $D$,
the answer is predicted by the model $M$.

\begin{algorithm}[tbp]
\caption{MLDP Method}\label{Alg-MLDP}
\begin{algorithmic}[tp]
\REQUIRE $F_{x}$, $F_{y}$ $\epsilon$, $D$.
\ENSURE $\widehat{F}(D)$.
\STATE 1. Measure $S_{F_{x}}$;
\STATE 2. $\widehat{F}_{x}(D)=F_{x}(D)+Laplace(S_{F_{x}}/\epsilon)$;
\STATE 3. Generate training set $T=<F_{x}, \widehat{F}_{x}(D)>$;
\STATE 4. Use regression algorithm or other prediction algorithms to generate model $M$;
\STATE 5. $\widehat{F}_{y}(D)=F_{y}(M)$;
\STATE 6. Output $\widehat{F}(D)$.
\end{algorithmic}
\end{algorithm}

\subsection{Training set Selection}

The performance of the model is affected by two types of errors.
One is \emph{noise error} $E_{N}$,
which is incurred by noise added to the training set.
Another is \emph{model error} $E_{M}$, which is triggered by the inaccuracy of the learning model.
According to the union bound,
the probability of the total error $E_{total}$ can be defined as
\begin{equation}\label{eq-etotal}
  Pr(E_{total})\leq Pr(E_{N})+Pr(E_{M})
\end{equation}

We define two criteria to measure the selection of training set.
One is \emph{independent},
which means how many queries are issued on one variable.
The independent is highly related to the sensitivity:
a high independent training set
leads to lower sensitivity.
Therefore,
when queries in training set are independent to others,
the noise error $E_{N}$ will be lower.

Another criteria is coverage,
which means how many variables can be covered by the training set.
It is obviously that if some variables cannot be covered by the training set,
the model error $E_{M}$ will be very high.
On the other hand,
if one query covers all variable,
the model error $E_{N}$ will still be very high as the model will be less fitting.

Therefore,
training set can be generated by those queries that have maximum coverage to the variables while with minimum
correlations.
Taking the dataset in Table~\ref{Tab-kbatch} as an example,
Table~\ref{Tab-kbatchquery} shows all possible query we can choose in the training set.
Queries ${f_{7}, f_{8}, f_{9}, f_{10}}$ can meet with the criteria.

%If all queries in the training set can cover all variables and independent to each other,
%the learning process is equal to the Laplace mechanism.
%
%The algorithm can also be complex, for example,
%we can choose different basis functions for linear regression.
%
%The number of parameter in the model is depend on the $m$,
%so with the different size of training set, the model is totally different.

\subsection{Training the Model and Making Prediction}

Model $M$ can be trained by various learning algorithms,
for example,
linear combinations of fixed nonlinear functions of the input variables can be used to train a model $M$.
\begin{equation}
F_{y}(M)=w_{0}+\sum_{i=1}^{d-1}w_{j}\phi_{j}(F_{x}),
\end{equation}
where $\phi_{j}(F_{x})$ is the Gaussian basis function, as shown in Equation~\ref{eq-gaussian}.

\begin{equation}\label{eq-gaussian}
\phi_{j}(F_{x})=\exp(-\frac{(\widehat{F}_x(D)-\mu)^2}{2u^2})
\end{equation}
where $\mu$ is the average value of $\widehat{F}_{x}(D)$
and $u$ is a pre-defined parameter to control the scalability of the basis function.
When model $M$ is generated,
queries, including fresh queries answers can be generated by $M$ without consuming
any privacy budget.

In fact,
the model depicts the combination of varies queries answers.
For example,
the linear regression model for range queries is approximate to a histogram publishing.
The parameters in a linear regression model are actually frequencies of a histogram.
But the linear regression model is not working well for the similarity query,
this is the because the combination of similarity queries are more complex comparing to range queries,
while a more sophisticated model is needed in this case, such as the neural network,
SVM models.
The effectiveness of these learning algorithms will be proved in the experiment.

\section{Analysis of the MLDP}

\subsection{Privacy Analysis of the MLDP Method}

According to the definition of differential privacy,
if the data processing follows the requirement
of differential privacy at each step,
the result will satisfy with differential privacy~\cite{Dwork1866758}.

Algorithm~\ref{Alg-MLDP}
shows that the privacy budget is only consumed in Step $2$,
in which \emph{Laplace} noise is added to answers of training queries $F_{x}$.
As the original dataset $D$ is only accessed by $F_{x}$,
following steps, model training and fresh query prediction,
will not disclose any privacy information.
Therefore,
we ensure every step in Algorithm~\ref{Alg-MLDP} satisfies differential privacy.

In addition,
the model $M$ will be published to the public to make prediction.
No matter what type of learning algorithms we choose to train the model,
it contains noise that will not release any privacy information.
Similar to other non-interactive data publishing methods,
once $M$ is published,
the privacy level of the model is fixed by $\epsilon$
that \emph{MLDP} consumes in the training step.
In general,
we have transferred the differentially private \emph{non-interactive data publishing}
problem into a \emph{machine learning} problem with the constraint of differential privacy.

\subsection{Utility Analysis of the MLDP Method} \label{sec-utility}

In this section,
we focus on the relationship between the size of the
training set and the accuracy of the query result.
The accuracy is determined by
a widely used utility notion in \emph{differential privacy}
suggested by Blum~{et~al.}~\cite{Blum2008LTA}:

\begin{definition}[($\alpha$,$\beta$)-useful]\label{DEF-USE}
A mechanism $\mathcal{M}$ is ($\alpha$,$\beta$)-useful
for a set of queries $F$ and a dataset $D$,
if:
with probability $1-\beta$,
for every query $f\in F$,
%for $\widehat{x} = \mathcal{M}(x)$,
we have
\begin{equation}
Pr(\max_{f\in F}|\widehat{F}(D)-F(D)|\leq\alpha)\geq 1-\beta,
\end{equation}
then $\alpha$ is the accuracy of the method and $\beta$ is the confidence parameter.
\end{definition}

Equation~\ref{eq-etotal} will help us to analyze the utility of \emph{MLDP}.
Based on the accuracy definition and the relationship between errors,
we will demonstrate that errors of \emph{MLDP} are bounded by a
certain value $\alpha$ with high probability.

%The accuracy of the proposed method is derived from the difference between
%$F(D)$ and $\widehat{F}(D)$, which consists of two types of error.
%One is \emph{noise error} $E_{N}$,
%which is incurred by noise added to the training set.
%Another is \emph{model error} $E_{M}$, which is triggered by the inaccuracy of the learning model.
%According to the union bound,
%the probability of the total error $E_{total}$ can be defined as
%\begin{equation}\label{eq-etotal}
%  Pr(E_{total})\leq Pr(E_{N})+Pr(E_{M})
%\end{equation}

The model error $E_{M}$ is associated with the type of query and the learning algorithm.
As \emph{MLDP} does not specify the query type,
we can use the range query as an example.
When a dataset has $n$ records,
the range query will output true values in a range from $0$ to $n$.
For other types of queries,
the difference is only in the range of query answers in the bound.
\emph{MLDP} also does not specify the learning algorithm.
We suppose the learning algorithm we choose has a hypothesis set $H=\{h_{1},...h_{i}\}$
with size $|H|$.
The accuracy bound of model error $E_{M}$ will be estimated by Theorem~\ref{thm1}
and will be proved with the help of the \emph{Chernoff-Hoeffding} bound~\cite{Kasiviswanathan2008531}.

\begin{lem}\label{lem2}
(Real-valued Chernoff-Hoeffding Bound)~\cite{Kasiviswanathan2008531}.
Let $X_{1}, ... , X_{m}$ be independent random variables with $E[X_{i}]=\mu$ and
$a\leq X_{i}\leq b$ for all i,
then for every $\alpha>0$,
\begin{equation}
Pr(|\frac{\sum_{i}X_{i}}{m}|>\alpha)\leq 2\exp(\frac{-2\alpha^2m}{(b-a)^2}).
\end{equation}
\end{lem}

\begin{thm}\label{thm1}
For any query $f\in F$,
all $\beta >0$,
with probability at least $1-\beta$,
the model error $E_{M}$ is bounded by
$\alpha\leq(\frac{n^2\ln(2|H|/\beta)}{2m})^{1/2}$,
where $|H|$ dependents on the learning algorithm we choose for MLDP.
\end{thm}

\begin{proof}
As
$E_{M}=\frac{\sum_{i}|f_{i}(D)-f(M)|}{m}$,
we have $$Pr(E_{M})=Pr(\frac{\sum_{i}|f_{i}(D)-f(M)|}{m})$$
and $$Pr(E_{M}>\alpha)=Pr(\frac{sum_{i}|f_{i}(D)-f(M)|}{m}>\alpha).$$
As range query result is a true value with the maximum value of $n$ and the
minimum value of $0$,
according to Lemma~\ref{lem2},
for every hypothesis $h\in H$, we have
$$Pr(\frac{sum_{i}|f_{i}(D)-f(M)}{m}|>\alpha)\leq 2\exp(-\frac{2m\alpha^2}{n^2}).$$
For all hypotheses, we then have
$$Pr(\frac{sum_{i}|f_{i}(D)-f(M)}{m}|>\alpha)\leq 2|H|\exp(-\frac{2m\alpha^2}{n^2})$$
Let
$\beta=2|H|\exp(-\frac{2m\alpha^2}{n^2})$,
We have
$\alpha\leq(\frac{n^2\ln(2|H|/\beta)}{2m})^{1/2}$.
\end{proof}

The noise error $E_{N}$ is independent to the
query type and the learning algorithm,
and can be analyzed by the property of \emph{Laplace} noise,
which is presented by sums of \emph{Laplace} random variables in Lemma~\ref{lem1}.
\begin{lem}\label{lem1}
(Sums of Laplace Random Variables)~\cite{Kasiviswanathan2008531}.
Let $\lambda_{1}, ... , \lambda_{m}$
be a set of independent random variables drawn from Laplace($\sigma$),
then for every $\alpha>0$,
\begin{equation}
Pr(|\frac{\sum \lambda_{i}}{m}|>\alpha)=\exp(-\frac{m\alpha^2}{4\sigma}).
\end{equation}
\end{lem}
In the \emph{MLDP} training set,
the noise $\lambda_{i}$ is derived from $Laplace(\frac{S_{T}}{\epsilon})$.
Theorem~\ref{thm2} shows the bound of the $E_{N}$.

\begin{thm}\label{thm2}
For any query $f\in F$,
all $\beta >0$,
with probability at least $1-\beta$,
the noise error $E_{N}$ of \emph{MLDP} is
bounded by
$\alpha=(\frac{4S_{T}\ln(|H|/\beta)}{m\epsilon^2})^{1/2}$,
where $|H|$ dependents on the learning algorithm we choose for the \emph{MLDP}.
\end{thm}

\begin{proof}
As
$E_{N}=\frac{\sum_{i}|\widehat{f}_{i}(D)-f_{i}(D)|}{m}$,
we have $$Pr(E_{N})=Pr(\frac{\sum_{i}|\widehat{f}_{i}(D)-f_{i}(D)}{m}|).$$
For each $f_{i}\in F$, \emph{MLDP} adds a random variable $\lambda_{i}$
as noise drawn from $Laplace(\frac{S_{F}}{\epsilon})$.
When $\alpha>0$,
we need to bound $Pr(E_{N}>\alpha)=Pr(|\frac{\sum\lambda_{i}}{m}|\geq\alpha)$.
According to Lemma 2,
we have
$$Pr(|\frac{\sum\lambda_{i}}{m}|\geq\alpha)=\exp(-\frac{m\alpha^2\epsilon^2}{4S_{F}}).$$
For all hypotheses $h\in H$, we then have
$$Pr(|\frac{\sum\lambda_{i}}{m}|\geq\alpha)=|H|\exp(-\frac{m\alpha^2\epsilon^2}{4S_{F}}).$$
Let
$\beta=|H|\exp(-\frac{m\alpha^2\epsilon^2}{4S_{F}})$,
we have
$\alpha=(\frac{4S_{F}\ln(|H|/\beta)}{m\epsilon^2})^{1/2}$.

\end{proof}

Theorem~\ref{thm1} illustrates the relationship between the accuracy and
the size of the training set.
The bound on $\alpha$ indicates that
a larger size of training set $m$ results in a lower $E_{M}$ and higher accuracy.
Theorem~\ref{thm2} shows that the accuracy is also associated with
the sensitivities $S_{T}$.
A larger $S_{T}$ leads to larger $E_{N}$ and lower accuracy.
As $S_{T}$ is dominated by the correlations between queries in the training set,
a larger $m$ will increase $S_{T}$.
Although a large $m$ increases the accuracy according to Theorem~\ref{thm1},
$S_{T}$ will be enlarged at the same
time which will impede the enhancing of the accuracy.
These results assist in selecting the training set of the \emph{MLDP}.

To retain acceptable accuracy,
we cannot use all possible queries to train the model.
At first glance,
it seems that if we can list all possible queries as the training set,
we can produce a very accurate model to make the prediction.
However, according to Theorem~\ref{thm2},
even though a large training set results in less model error $E_{M}$,
a large training set contains huge number of correlated queries,
which leads to higher sensitivity and large noise error $E_{N}$.
The total error $E_{total}$ will be quite large and the accuracy will be reduced.

We cannot use only the uncorrelated queries as the training set.
In this case,
$E_{N}$ can be quite small due to the smaller sensitivity and lower volume of noise;
however;
as uncorrelated queries are only a small proportion of the entire number of queries,
the model will be very inaccurate and $E_{M}$ will be very large.
Both cases indicate that we must consider two types of error
to achieve better performance.

\subsubsection{Advantage of MLDP}

There are several advantages to using a machine learning method to
deal with the non-interactive publishing problem:
\begin{itemize}
\item
Many machine learning models can be applied to the data publishing problem.
For example, we can use \emph{linear regression},
\emph{SVM for regression} and \emph{neural network} to learn a
suitable model $M$ according to the training set.
As prediction is a mature area that has been investigated in machine learning for several decades,
we can choose sophisticated technologies and adjust parameters to obtain a better performance.
%Lots of machine learning theory can be applied in this process.

\item Some existing methods can be considered as an extension of \emph{MLDP}.
For example,
the Private Multiplicative Weights (PMW) Mechanism~\cite{Hardt201061}
is one of the most prevalent publishing methods in differential privacy.
To some extent,
it can be considered as an instance of the \emph{MLDP} method.
In the \emph{PMW},
the histogram is a selected model and frequencies in this histogram
constitute parameters of the model.
The model (histogram) is trained by the input queries until
it converges or meets the halting criteria.
Compared to MLDP, however,
\emph{PMW} can only answer queries in the training set.

\item
The noisy model naturally has the property of generalization.
Generalization is an essential problem in machine learning, but the differential privacy mechanism
has proven that it can avoid over-fitting in the learning process~\cite{Dwork2015validation}.

%\item
%The model can be modified to cater the requirement of the various privacy level.
%If curator intends to preserve a higher privacy level,
%he can limited the training set.

\end{itemize}

\subsection{Differences between MLDP and Private Learning}

The proposed \emph{MLDP} method is highly related to machine learning algorithms,
but is different to previous private learning.
It is a method that introduces noise into the original learning algorithms,
so that the privacy of the training dataset can be preserved in the learning process~\cite{Kasiviswanathan2008531}.
First, the purpose of publishing models is different.
\emph{MLDP} aim to publish a model for fresh query prediction,
whereas private learning is only used for traditional machine learning
tasks and will not preserve the privacy of fresh samples.

Second,
\emph{MLDP} considers pre-defined queries as a training set
while private learning considers records in the original dataset as traninig samples.
The target of differential privacy is to hide the true value of query answers,
not the records,
so \emph{MLDP} considers the query as the training sample and
the model is used to predict query answers rather than the values of records.
In this respect,
\emph{MLDP} is totally different from private learning algorithms.
Even though Kasiviswanathan~\cite{Kasiviswanathan2008531}
proved that Kearn's statistical query (SQ)~\cite{Kearns1998}
model can be implemented in a differentially private manner,
the training set still comprises records in the dataset.
Consequently,
the \emph{SQ} model is similar to private learning, not \emph{MLDP}.

Finally,
as public users normally use count, average or sum query,
\emph{MLDP} normally applies regression algorithms for true value prediction,
while the private learning algorithm is usually specific to
classification with labels of $0$ or $1$.
Table~\ref{tab-difference} summarizes the major
differences between \emph{MLDP} and private learning.

\begin{table*}  \centering
  \caption{The difference between the MLDP and Private Learning}\label{tab-difference}
  \begin{tabular}{p{1.5in}|p{2.2in}|p{2.2in}}
\toprule
    & \emph{MLDP} & Private Learning \\
\midrule
    Model
    & The model is used to predict fresh query answers for public users
    & The model is used for traditional machine learning \\\midrule
    Training set
    & Queries on the dataset
    & Records in the dataset \\\midrule
    Protect Target
    & Preserve the privacy of all queries.
    & Do not protect future samples \\\midrule
    Learning Algorithms
    & Prediction
    & Classification \\\midrule
\bottomrule
\end{tabular}
\end{table*}

\section{Experiment and Analysis} \label{sec-experiment}
%This section evaluates the performance of
%the proposed \emph{MLDP} by answering the following questions:
%
%\begin{itemize}
%
%  \item \emph{How does the size of the training set impact on the performance of MLDP?}
%
%  The proposed \emph{MLDP} method aims to publish a model to predict a large set of fresh queries.
%  According to the utility analysis in Section 3,
%  the accuracy of the model is determined by the size of the training set.
%  Sub-section~\ref{sub-train} will investigate \emph{MLDP}'s performance
%  in terms of \emph{Mean Absolute Error} (MAE) on diverse sizes of training set,
%  and will compare it with the \emph{Laplace} method.
%
%  \item \emph{How does the size of the test set impact on the performance of MLDP?}
%
%  The size of the test set is the query number that \emph{MLDP} can answer in a batch.
%  In sub-section~\ref{sub-test},
%  we will explore the performance of MLDP on different numbers of queries
%  and compare it with the \emph{Laplace} method.
%
%
%  \item \emph{In which situation can MLDP can outperform other methods?}
%
%  Sub-section~\ref{sub-alg} will compare \emph{MLDP} with two other prevalent methods.
%  It will demonstrate the effectiveness of MLDP when answering large sets of queries.
%\end{itemize}

\subsection{Experiment Configuration}
The experiments involve four datasets.
Three are derived from Hay's work~\cite{hay2010boosting},
and have been widely used in the differentially private publishing tests.

\begin{itemize}
\item\texttt{NetTrace}:
this dataset contains the IP-level network trace at a border gateway of a university.
Each record reports the number of external hosts connected to an internal host.
There are $65,536$ records with the number of connections ranging from $1$ to $1423$.

\item\texttt{Search Logs}:
this synthetic dataset was generated by
interpolating \emph{Google Trends} data and \emph{America Online} search logs.
It contains $32,768$ search records collected
between \emph{Jan. 1, 2004} and \emph{Aug. 9, 2009}.

\item\texttt{Social Network}:
this dataset records the friendship relations among $11,000$ students,
sampled from an online social network website.
There were $32,768$ students,
each of which had at most $1678$ friends.

\item\texttt{Netflix}
The \emph{Netflix} dataset is extracted from the \emph{Netflix} Prize dataset,
where each user rated at least $20$ movies,
and each movie was rated by $20-250$ users.
This dataset is used to test the similarity query.

\item\texttt{Simulated Histogram}:
this simulated histogram contains $10$ bins with random numbers.
The target of  the simulated histogram is to list all
possible queries.
For a histogram with $10$ queries, the range query number is $1023$.
We will test all the queries in the experiment.

\item\texttt{Netflix}:
the \emph{Netflix} dataset is extracted from the \emph{Netflix} Prize dataset,
where each user rated at least $20$ movies,
and each movie was rated by $20-250$ users.
This dataset is used to test similarity queries.
\end{itemize}

We select different learning algorithms, including linear regression,
neutral network, ensemble bag, boost, and SVM, to create prediction model.
Two types of queries will be tested based on the prediction model:
\emph{range query} and \emph{similarity query}.
Range query is normally presented as \emph{how many users are there in the dataset with the age from 20 to 40?}
Range query is actually a count related query, which count how many record that meet with a specified property.
Similarity query is a type of complex query that measures the similarity between two records.
For example, \emph{what is the similarity between those two users in terms of their preference on movies?}
For both range and similarity queries,
>>>>>>> .r2804
%Each dataset is transferred to a histogram representation and \emph{MLDP} takes the range query as a query example.
%\emph{MLDP} can also be implemented in other type of queries.
we generated a training set with $m=1,000$ random queries $F$.
These range queries are correlated to one another and the sensitivity
is measured by definition~\ref{def-SG}.
The accuracy of results was measured by \emph{Mean Absolute Error} (MAE).
\begin{equation}\label{eq-mse}
  MAE=\frac{1}{m}\sum_{F_{i}\in F}|\widehat{F}_{i}(D)-F_{i}(D)|
\end{equation}
A lower \emph{MAE} implies better utility.

\subsection{Model Selection}
To select the model that fits for a particular type of query,
several algorithms have been tested with $\epsilon=1$,
including regression, SVM, neural network, bagging, and boosting.
Figure~\ref{FIG-Model} shows the results on Social Network and Netflix datasets.
Figure~\ref{Fig-Social-Model} illustrates that regression algorithm outperforms other
algorithms in terms of the range query.
When the size of the training set increases,
the MAE of the regression decreases dramatically.
When the size of the training set increases to more than $400$,
MAEs of all algorithm are stable.
These mean that the regression algorithm can be fully trained in limited training samples,
while others need more training samples.

However, for similarity queries,
the performance of regression is worse than others.
As shown in  Figure~\ref{Fig-Netflix-Model},
the SVM algorithm has the lowest MAE comparing with other algorithms.
This is because the similarity query
needs complex combination that cannot be deduced easily.
The SVM algorithm is more suitable to simulate the combination of queries.
Due to limited space,
we only used range query and regression model as examples in following set of experiments.

\begin{figure}[htbp]
\centering
\subfloat[\texttt{Range query}]{
\label{Fig-Social-Model}
\includegraphics[scale=0.23]{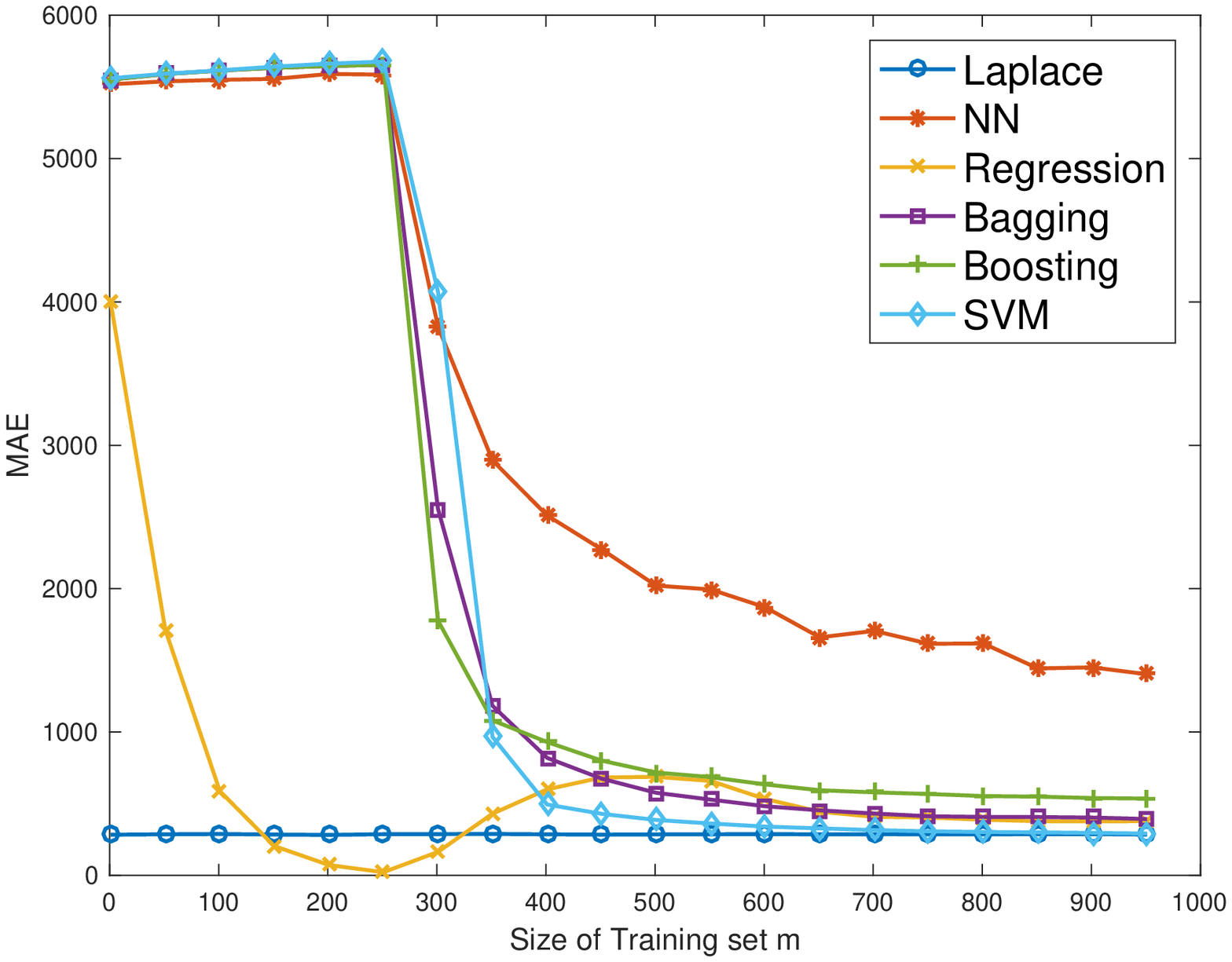}
}
\subfloat[\texttt{Similarity query}]{
\label{Fig-Netflix-Model}
\includegraphics[scale=0.23]{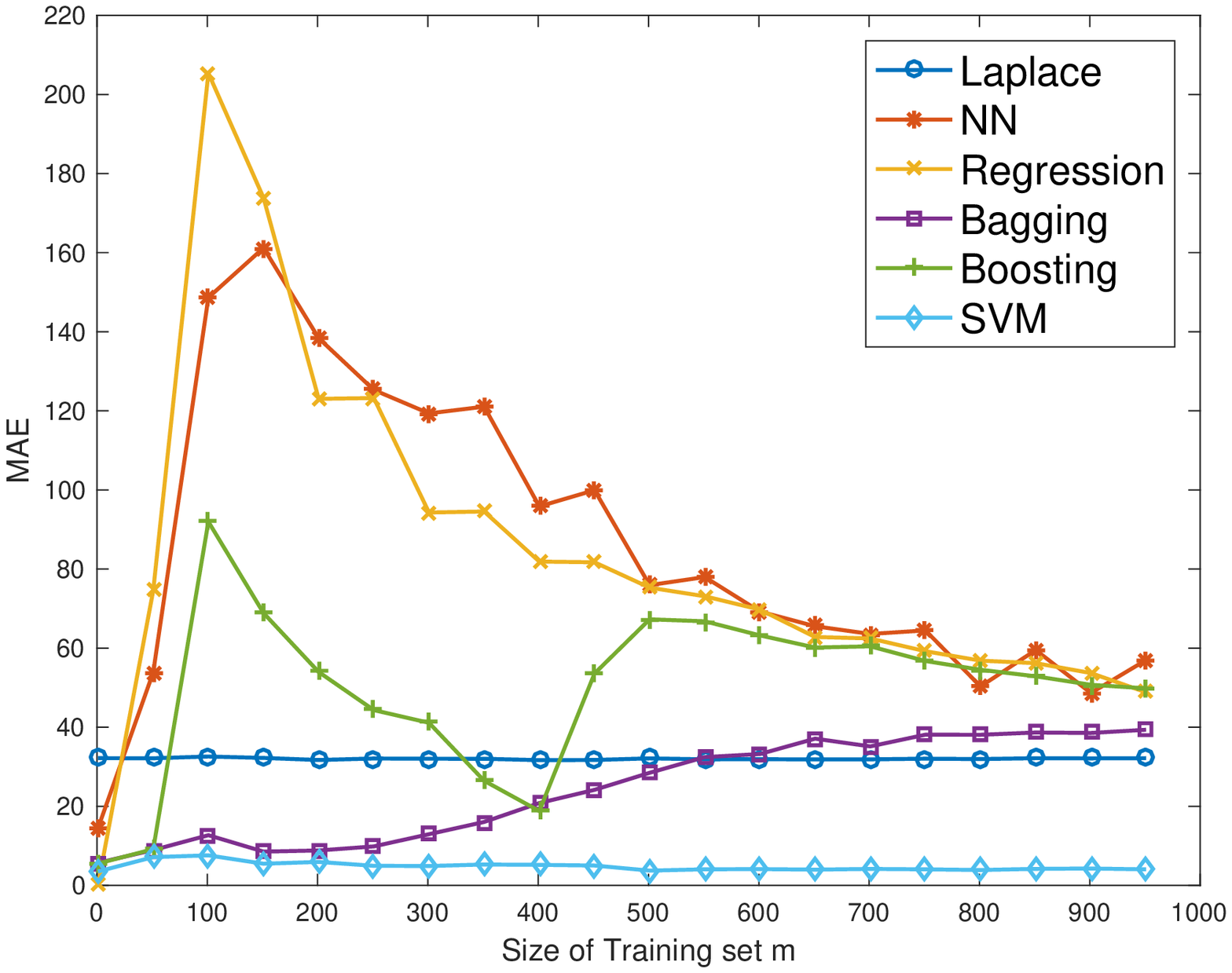}
}
\caption{Model selection on two different type of queries}
\label{FIG-Model}
\end{figure}

\subsection{Performance of MLDP with Different Sizes of Training Set}\label{sub-train}

\begin{figure*}[htbp]
\centering
\subfloat[\texttt{Nettrace}]{
\label{Fig-Nettrace-train-100test}
\includegraphics[scale=0.29]{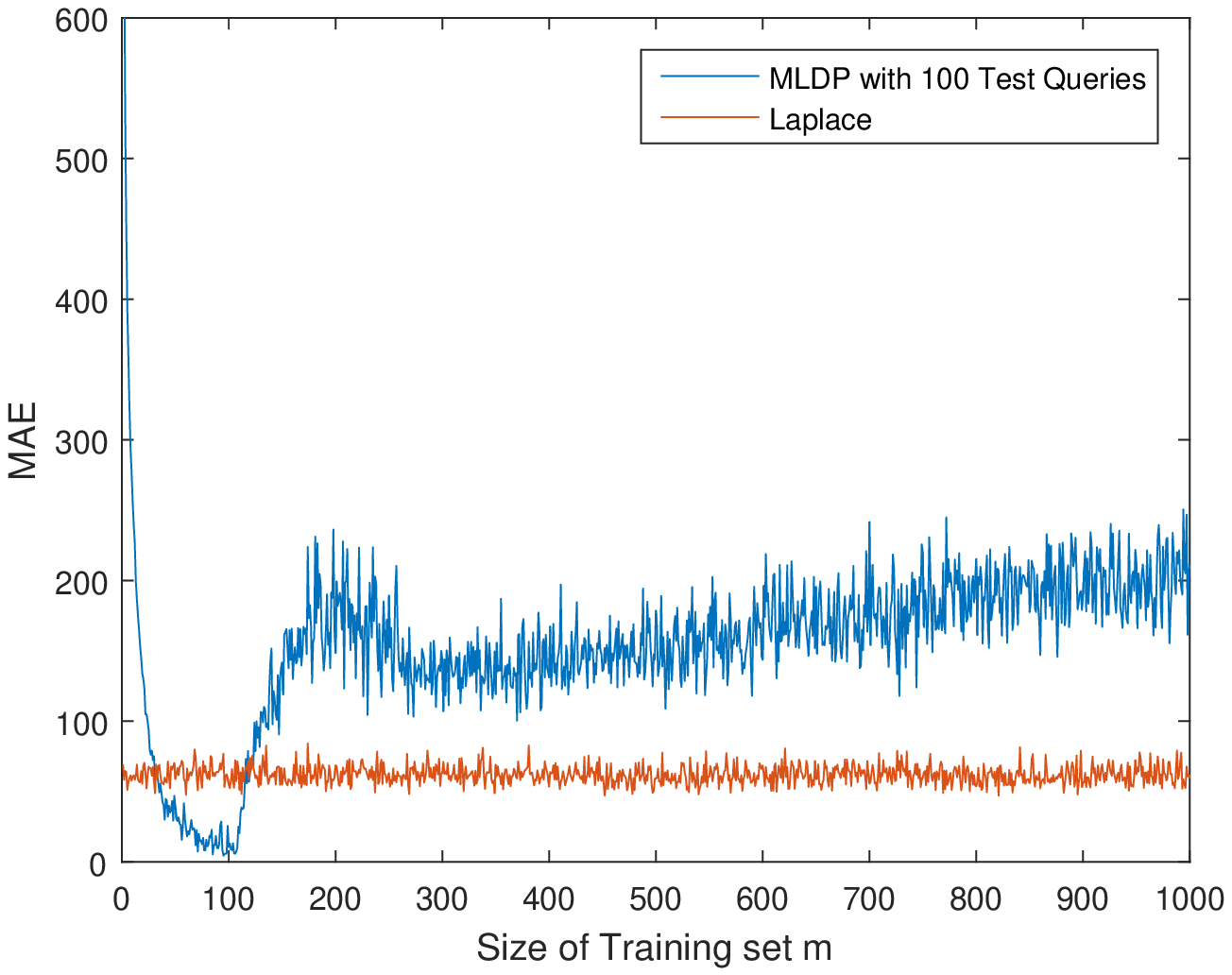}
}
\subfloat[\texttt{Nettrace}]{
\label{Fig-Nettrace-train-500test}
\includegraphics[scale=0.29]{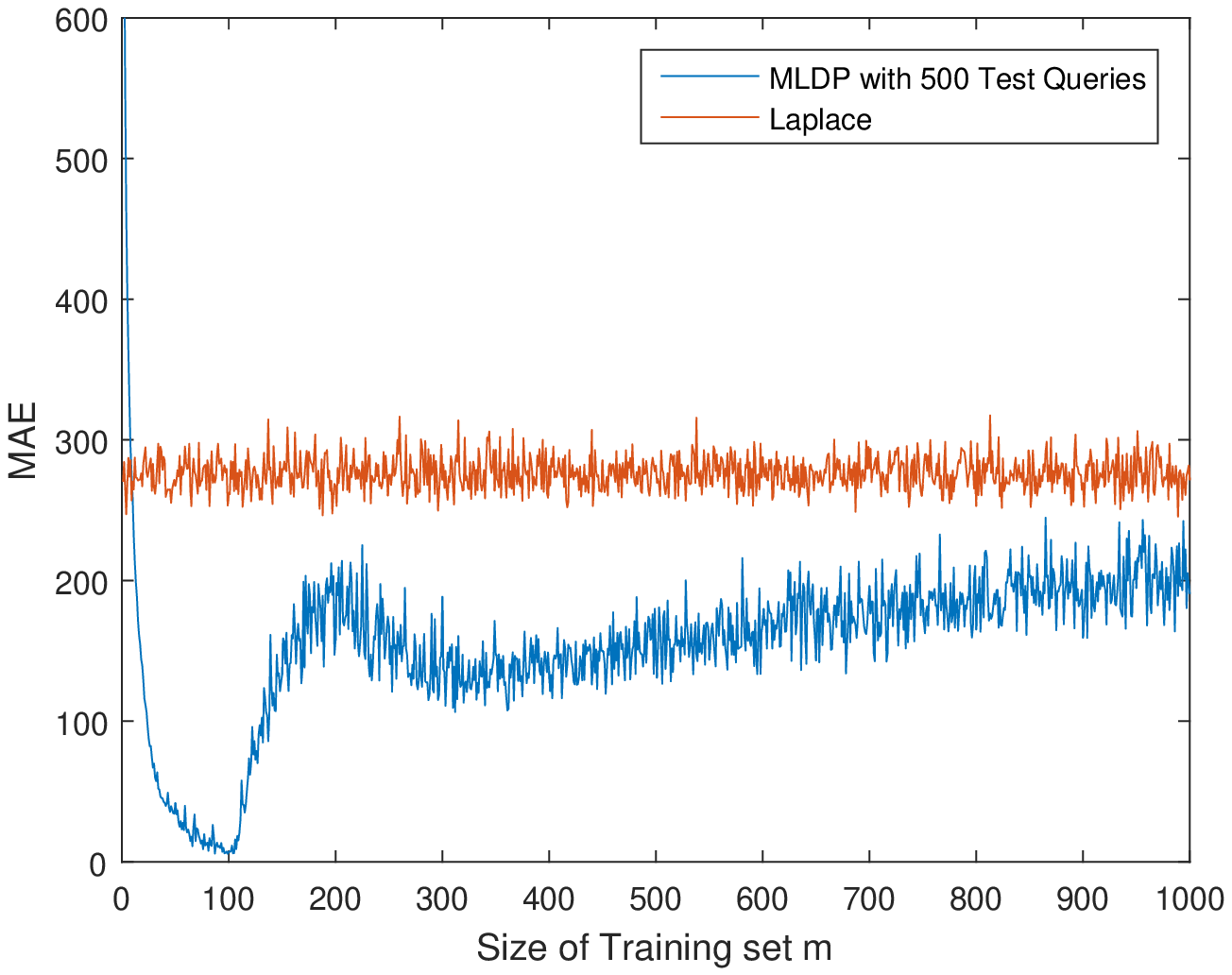}
}
\subfloat[\texttt{Search Log}]{
\label{Fig-Search-train-100test}
\includegraphics[scale=0.29]{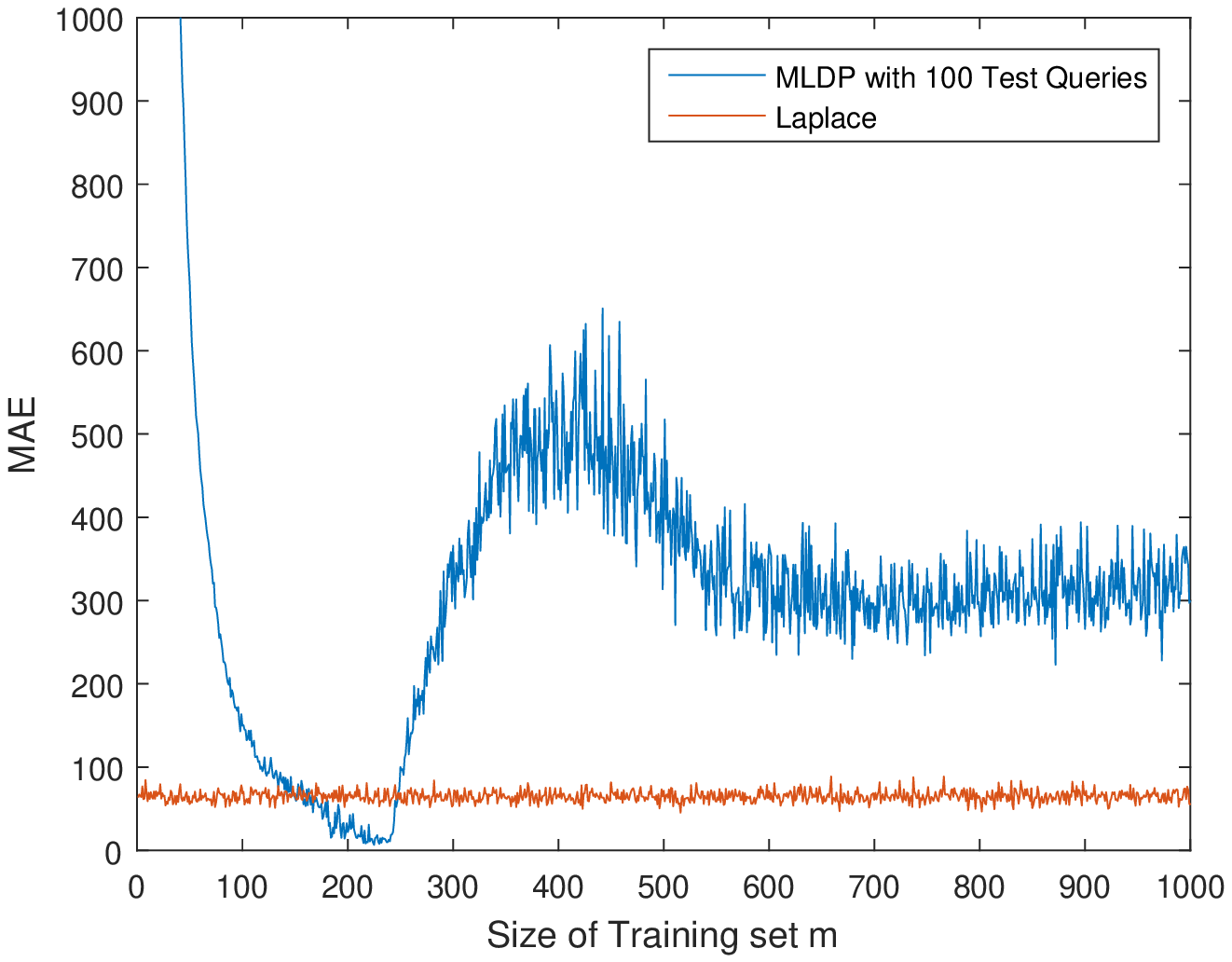}
}
\subfloat[\texttt{Search Log}]{
\label{Fig-Search-train-500test}
\includegraphics[scale=0.29]{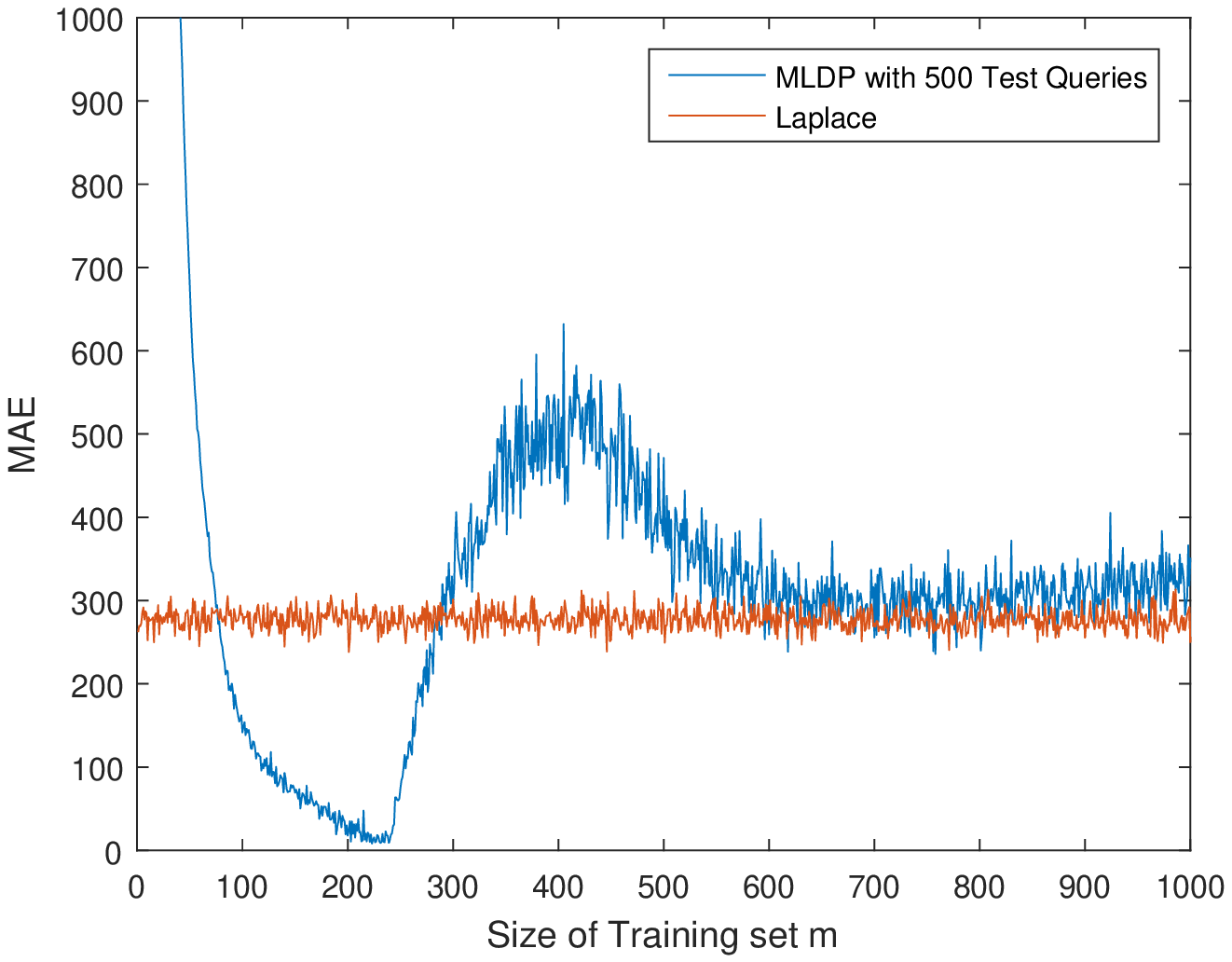}
}

\subfloat[\texttt{Social Network}]{
\label{Fig-Social-train-100test}
\includegraphics[scale=0.29]{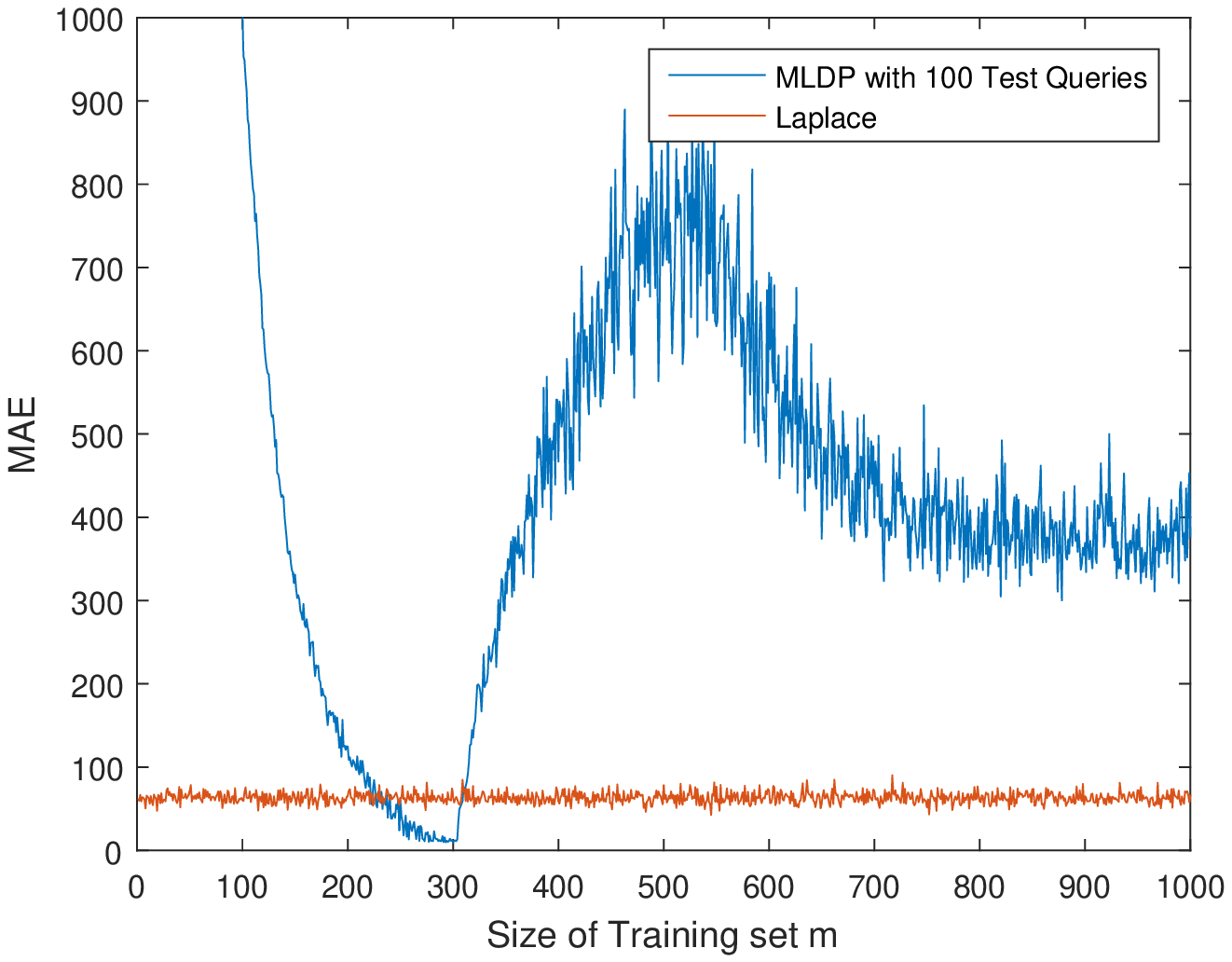}
}
\subfloat[\texttt{Social Network}]{
\label{Fig-Social-train-500test}
\includegraphics[scale=0.29]{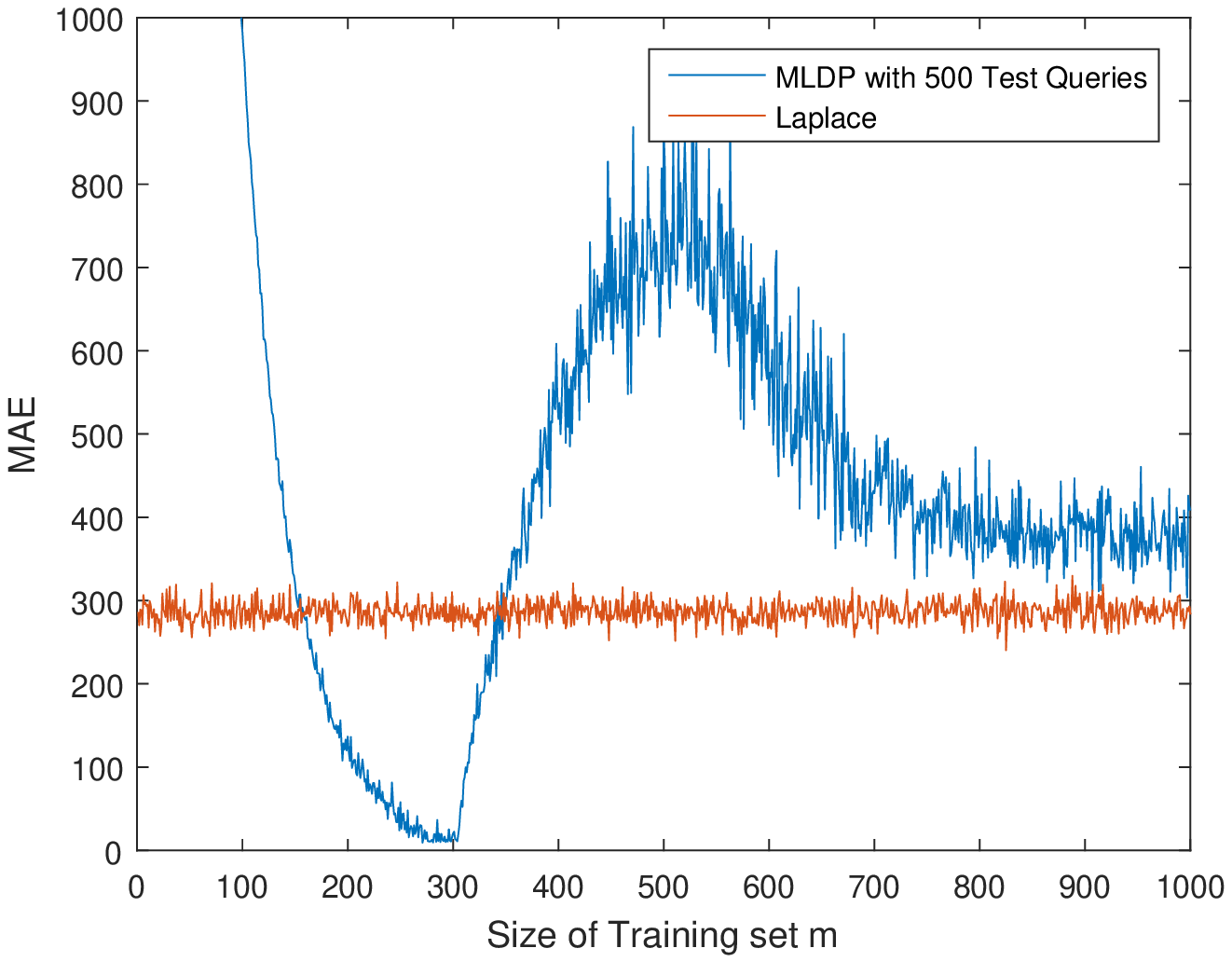}
}
\subfloat[\texttt{Simulated Histogram}]{
\label{Fig-Simu-train-100test}
\includegraphics[scale=0.29]{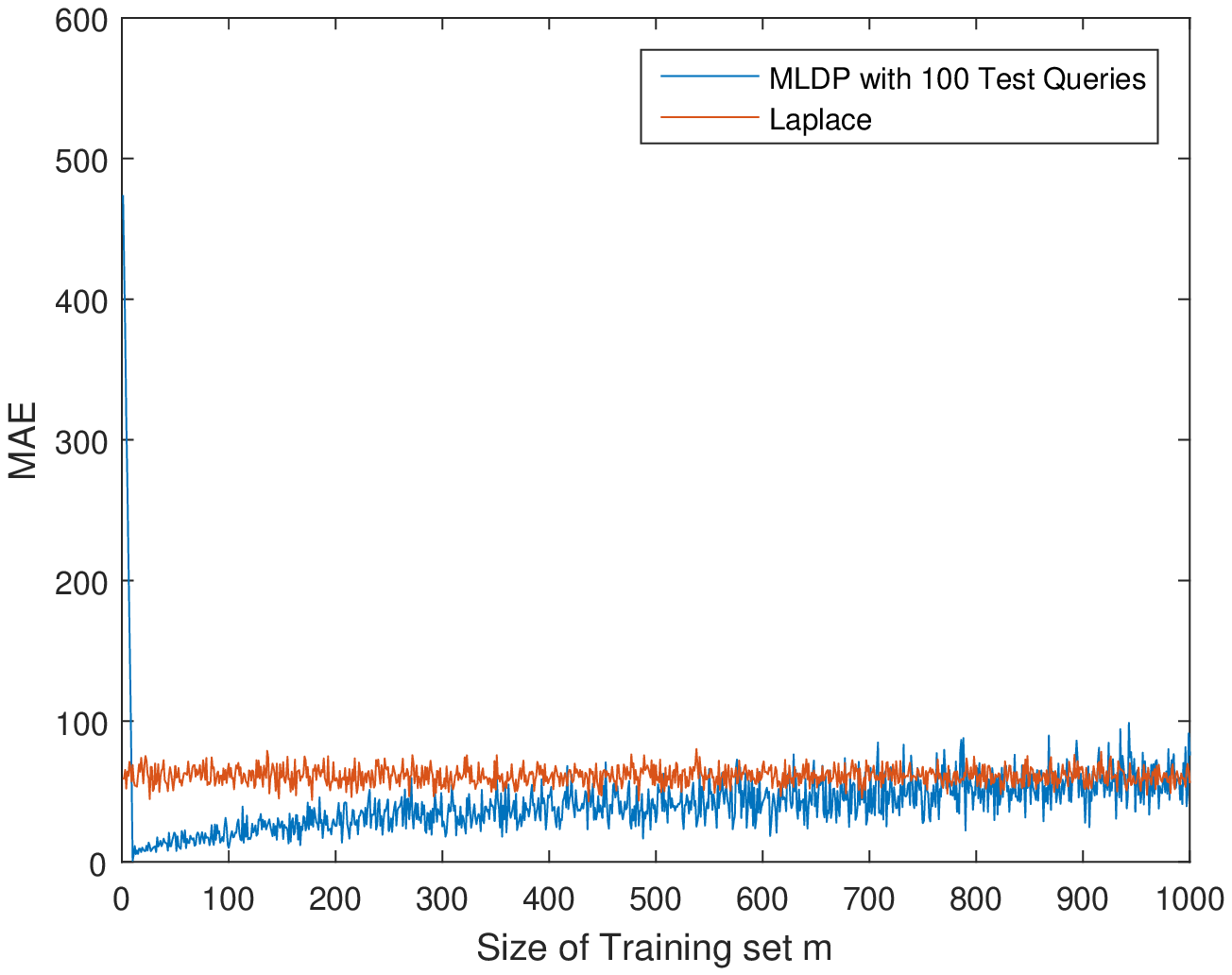}
}
\subfloat[\texttt{Simulated Histogram}]{
\label{Fig-Simu-train-500test}
\includegraphics[scale=0.29]{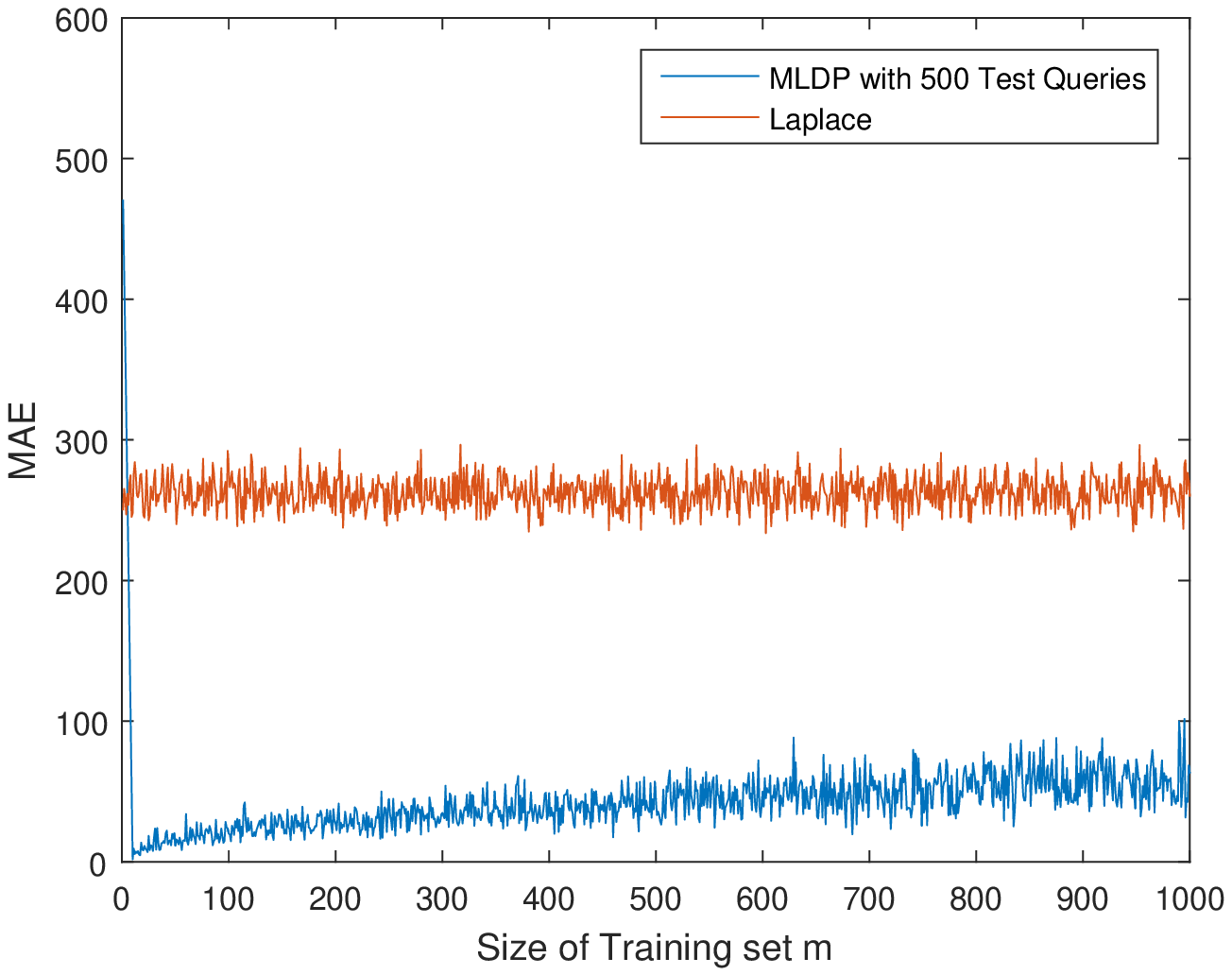}
}
\caption{Performance of \emph{MLDP} with different sizes of training set $\epsilon$}
\label{FIG-Train}
\end{figure*}

The theory analysis in Section $3$ indicates the size of the training set
plays a vital role in prediction.
This experiment examines various sizes of training set and
test set.
We set the size of training sets $m$ from $1$ to $1000$,
and set two different sizes of test query,
$m_{t}=100$ and $m_{t}=500$, respectively.
The $MAE$ result is compared with the traditional \emph{Laplace} method and the $\epsilon$ is fixed at $1$.

Fig.~\ref{FIG-Train} shows the impact of the size of the training set on the performance of \emph{MLDP}.
Initially,
it is apparent that \emph{MAE} drops quickly with the increase in $m$,
but when $m$ is larger than a certain value,
the \emph{MAE} reaches its minimum and
continues to increase.
As shown in Fig.~\ref{Fig-Nettrace-train-100test},
the \emph{MAE} keeps decreasing until $m=88$,
with $\emph{MAE}=20.3104$ at its lowest point.
As $m$ subsequently increases,
the \emph{MAE} keeps rising until $m=212$,
until it reaches $223.9952$.
The \emph{MAE} then goes down slowly as $m$ increases.
When it reaches another inflexion,
the \emph{MAE} slowly rises.
This result is consistent with Theorem~\ref{thm1} and~\ref{thm2} in Section~\ref{sec-utility},
in which the performance of \emph{MLDP} is impacted
by the mixture of noise error and model error.
The model error plays a dominate role when $m$ is small, so the \emph{MAE} decreases with $m$ increasing.
Beyond the threshold, however,
the noise error dominates the results.
As a larger $m$ introduces a larger volume of noise,
so the \emph{MAE} is enhanced.

Similarly, Fig.~\ref{Fig-Search-train-100test} shows that
the \emph{MAE} reaches its minimum when $m=227$.
Fig.~\ref{Fig-Social-train-100test} and Fig.~\ref{Fig-Simu-train-100test} show the same trend.
These results also confirm the theoretical analysis in sub-section~\ref{sec-utility}:
even a large $m$ can increase the accuracy of the model,
but too large a $m$ incurs a significant amount of noise in training sets and reduces the utility of the model.
These results illustrate the relationship between
$m$ and the performance of the \emph{MLDP},
which helps us to control the training set selection in \emph{MLDP}.

This set of experiments also compares the performance of \emph{MLDP}
with that of the \emph{Laplace} method in $m_{t}=100$ and $m_{t}=500$,
in which both methods will publish $100$ and $500$ queries in a batch.
Fig.~\ref{Fig-Nettrace-train-100test} and Fig~\ref{Fig-Nettrace-train-500test} shows that the
\emph{MAE}s of \emph{MLDP} in both figures are at the same,
which shows that the performance of \emph{MLDP} has not been impacted by the size of the test sets.
This is because when the model is fixed,
the prediction results will not change.
However,
the size of the test sets has great impact on the performance of the \emph{Laplace} method.
When $m_{t}=100$, the average \emph{MAE} is $61.8209$,
while the average \emph{MAE} increases to $275.5021$ when $m_{t}=500$.
This is because the total sensitivity increases with the enhancement of the
test set.
The noise added to each query is enlarged accordingly.
This trend is consistent in other datasets.
These results prove that when publishing large set of queries,
\emph{MLDP} significantly outperforms the traditional \emph{Laplace} method.

\subsection{Performance of MLDP with Different Sizes of Test Set}\label{sub-test}

\begin{figure*}[htbp]
\centering
\subfloat[\texttt{Nettrace}]{
\label{Fig-Nettrace-test-100train}
\includegraphics[scale=0.29]{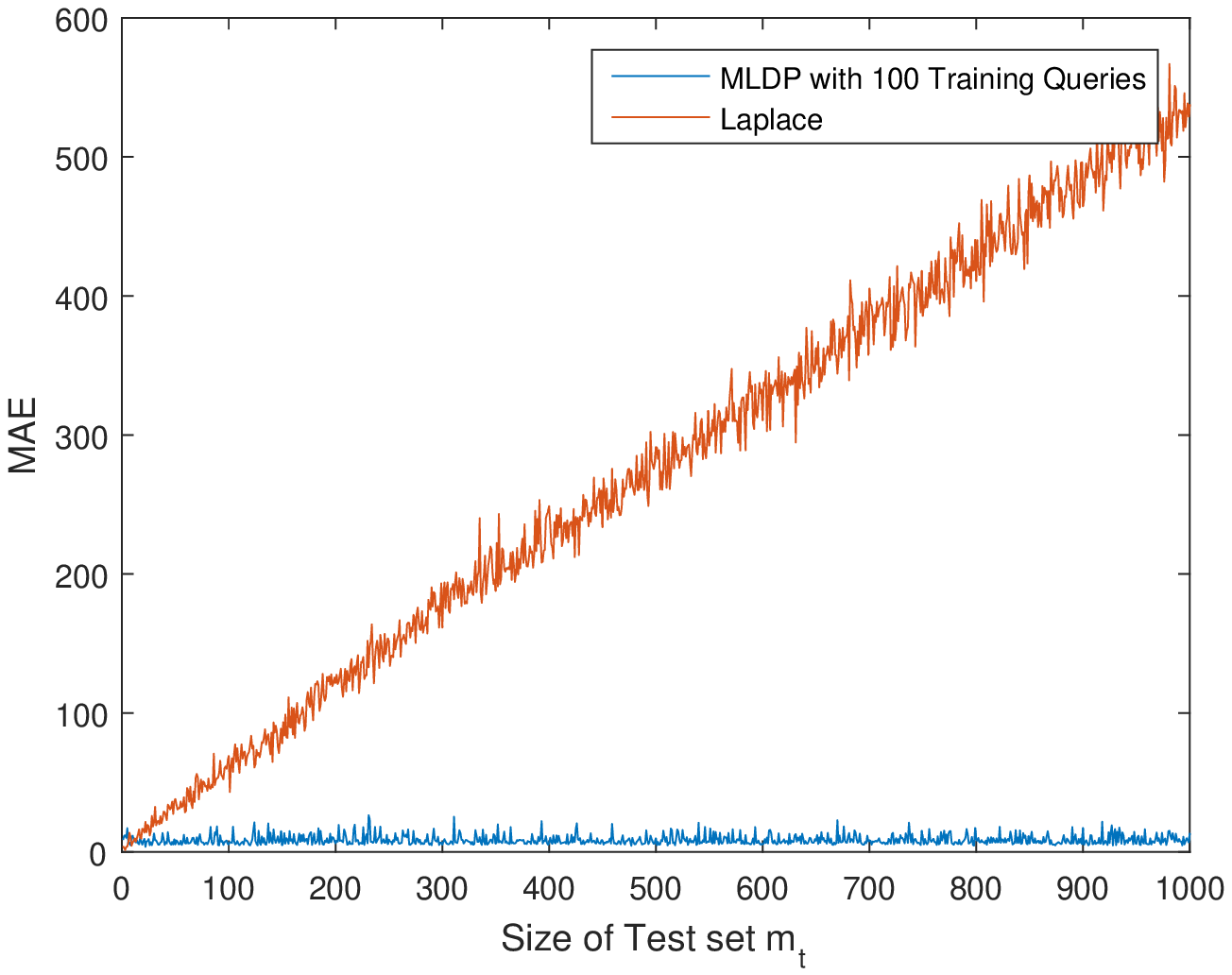}
}
\subfloat[\texttt{Nettrace}]{
\label{Fig-Nettrace-test-200train}
\includegraphics[scale=0.29]{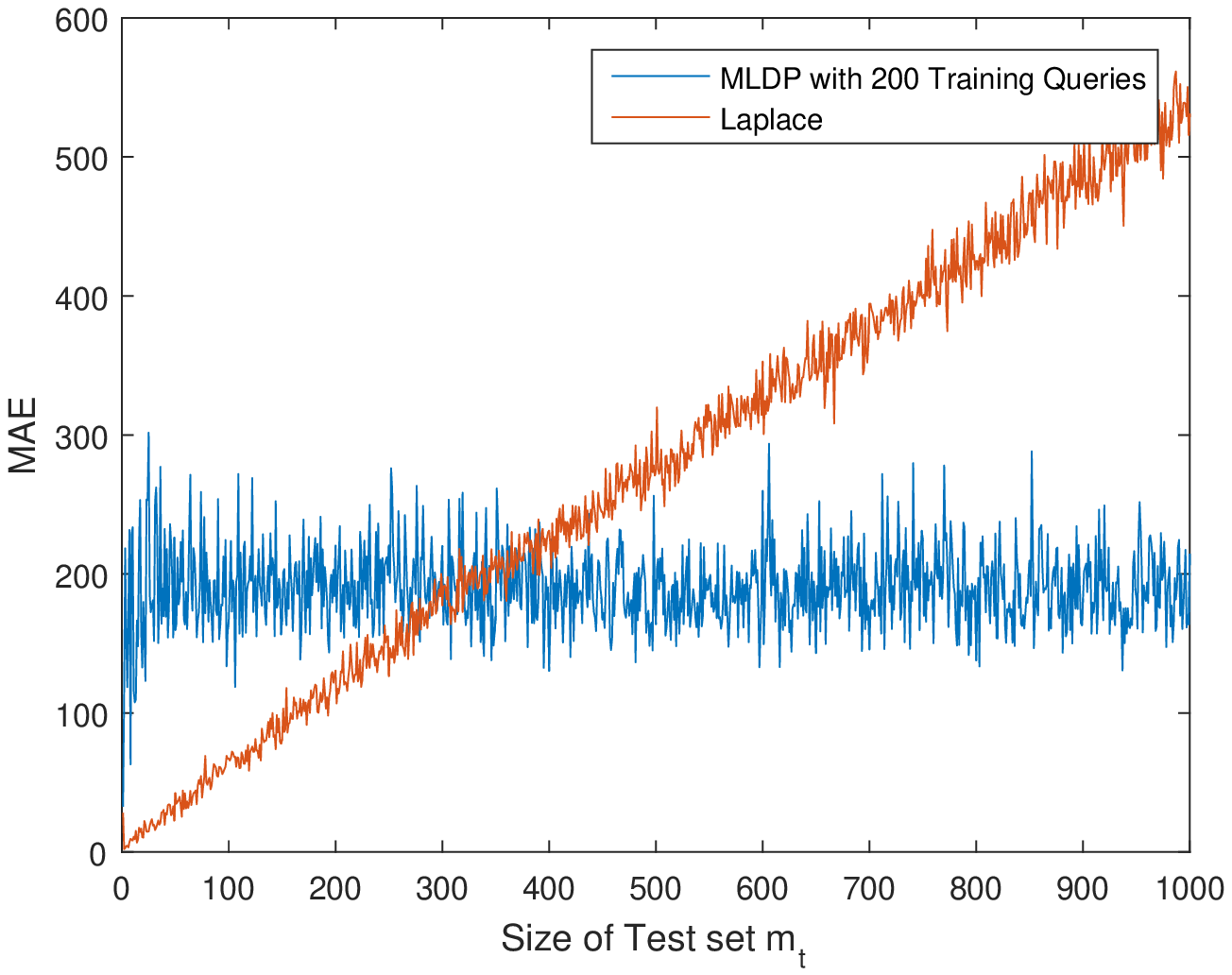}
}
\subfloat[\texttt{Search Log}]{
\label{Fig-Search-test-100train}
\includegraphics[scale=0.29]{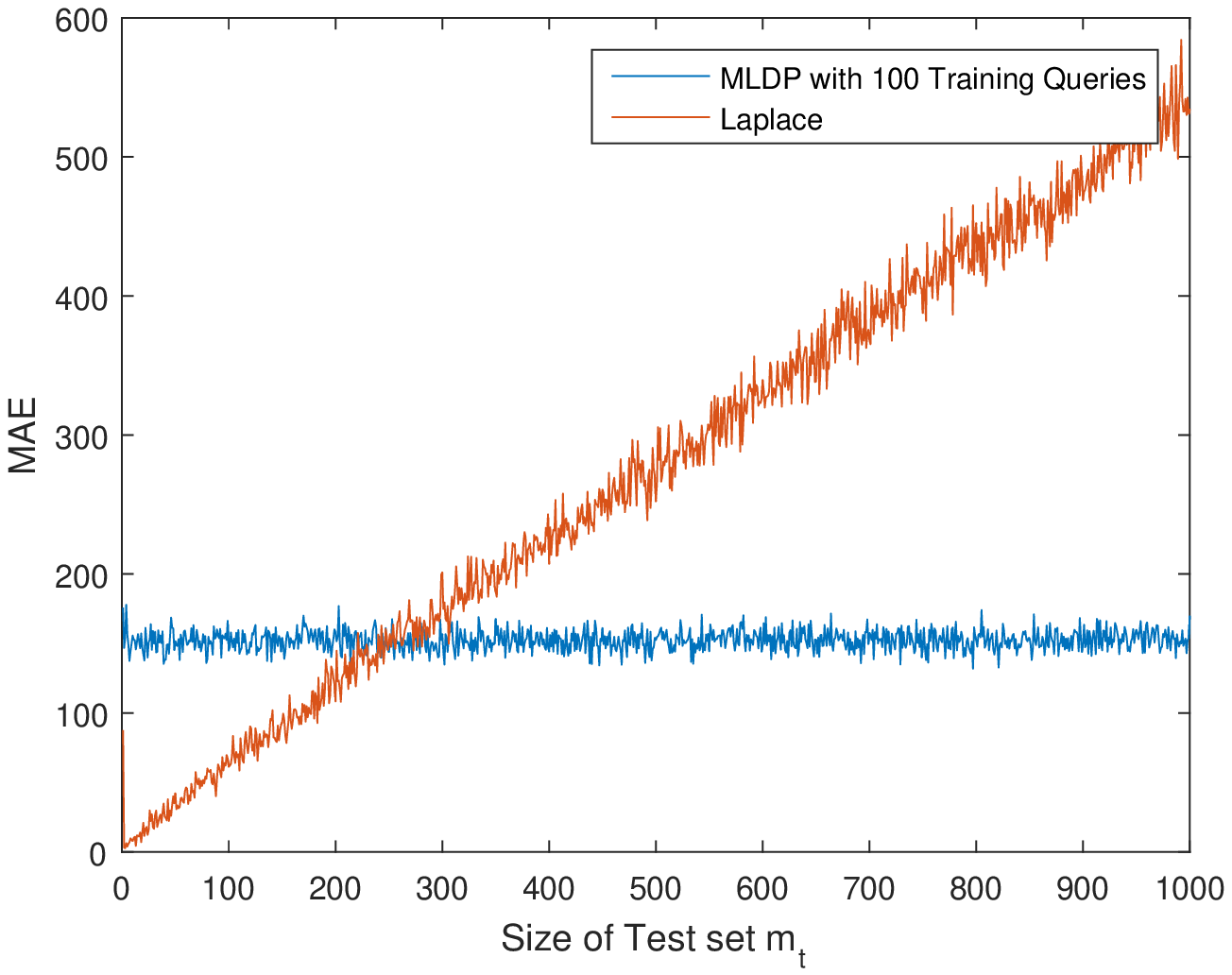}
}
\subfloat[\texttt{Search Log} with $500$ Test queries]{
\label{Fig-Search-test-200test}
\includegraphics[scale=0.29]{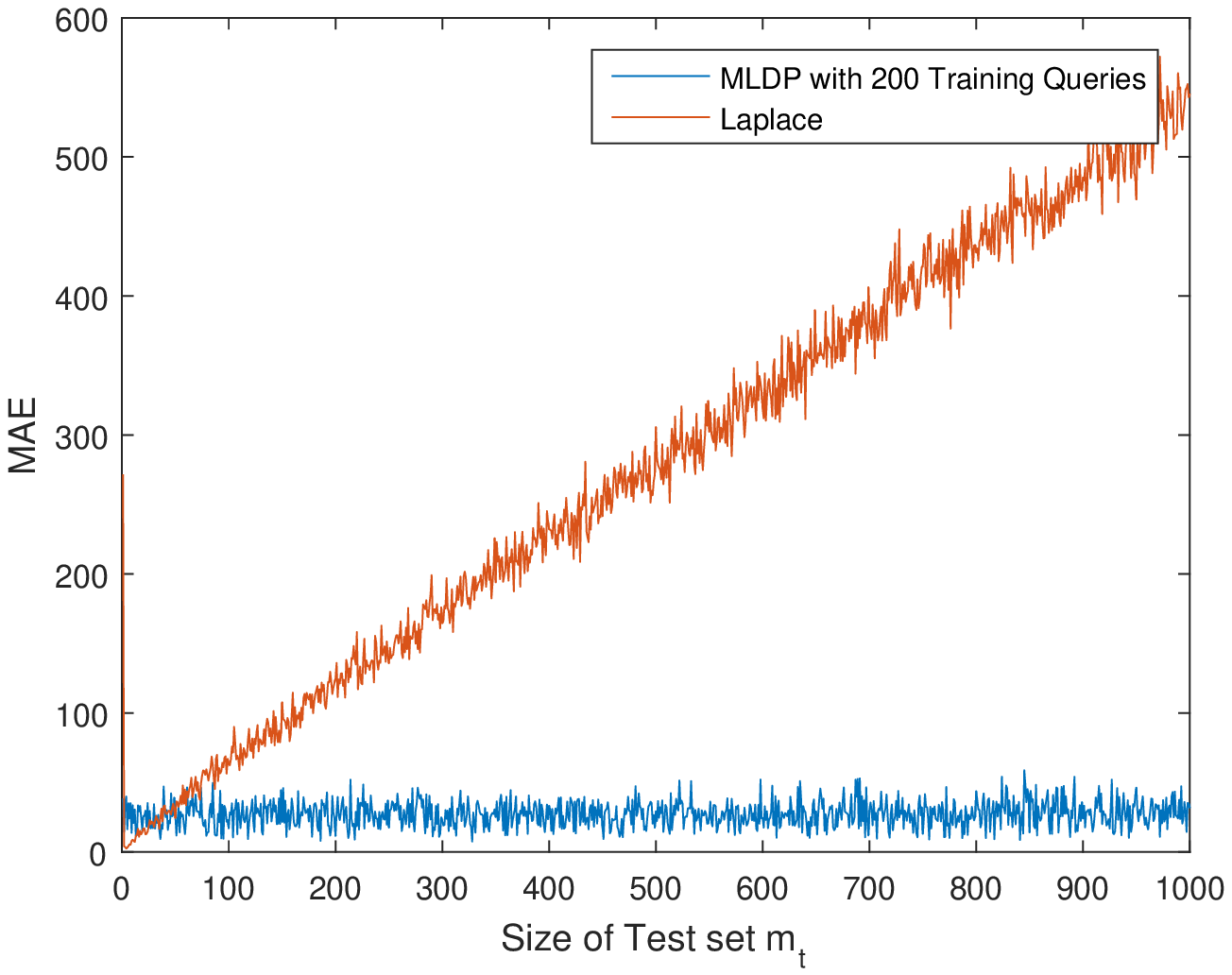}
}

\subfloat[\texttt{Social Network}]{
\label{Fig-Social-test-200train}
\includegraphics[scale=0.29]{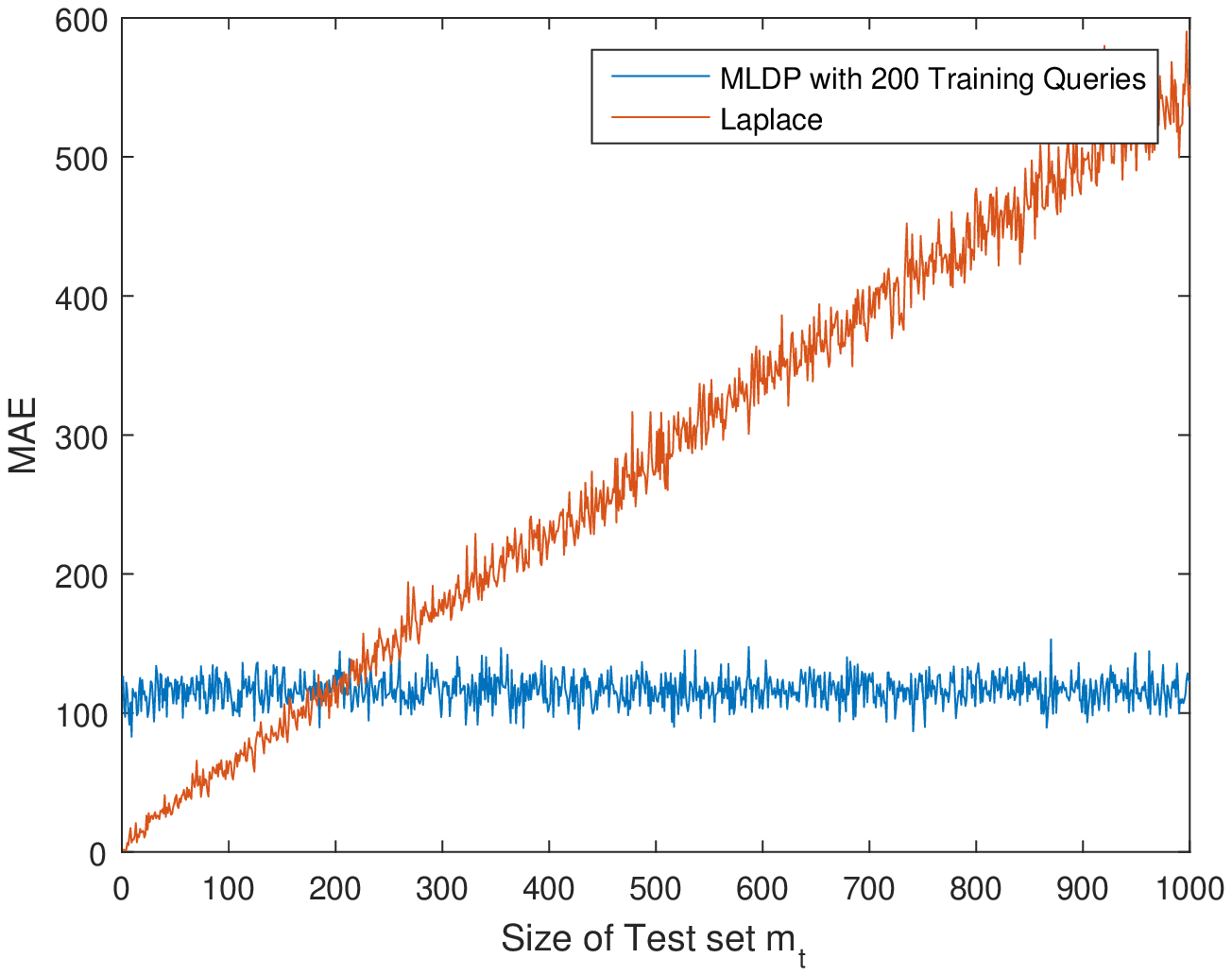}
}
\subfloat[\texttt{Social Network}]{
\label{Fig-Social-test-300train}
\includegraphics[scale=0.29]{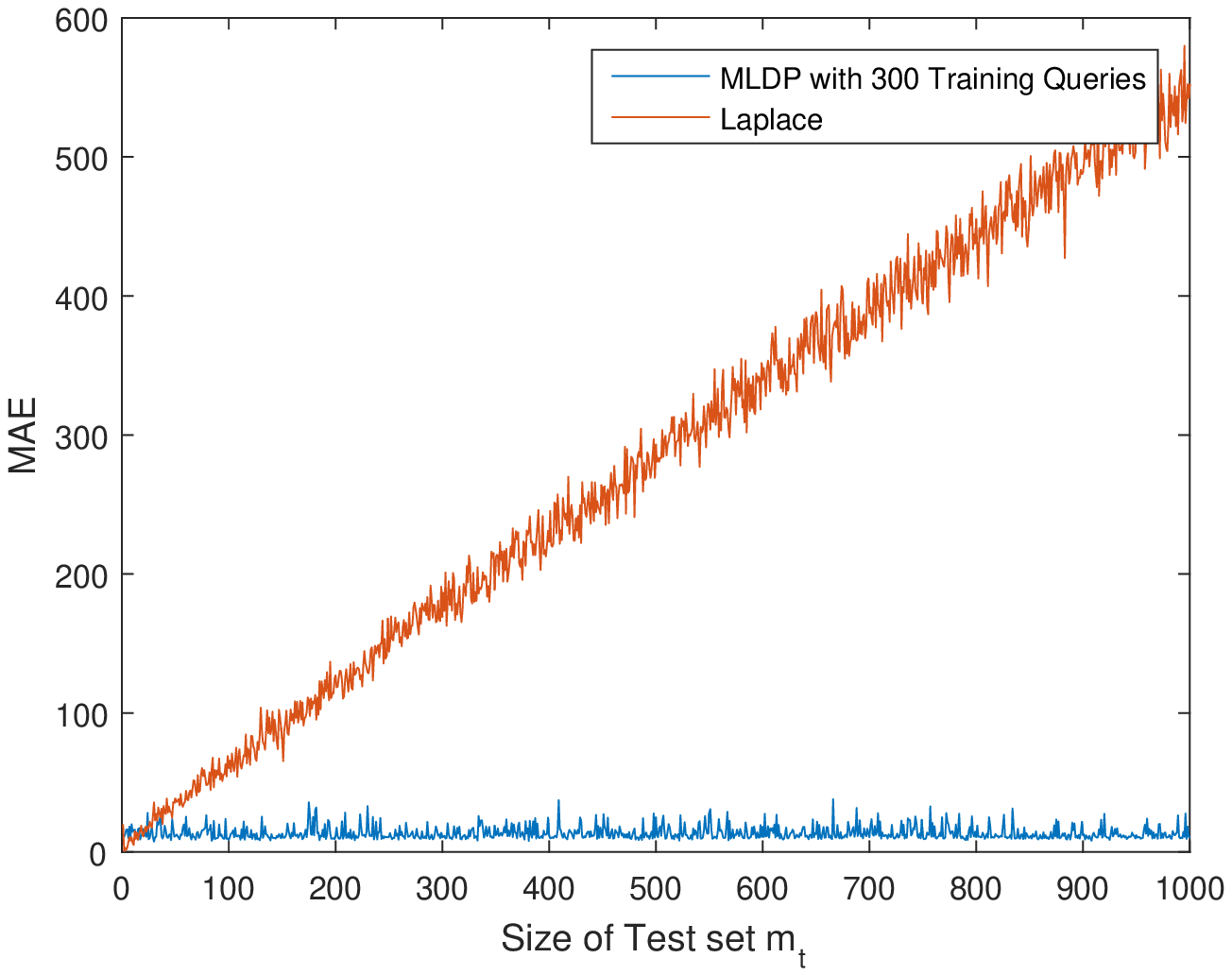}
}
\subfloat[\texttt{Simulated Histogram}]{
\label{Fig-Simu-test-100train}
\includegraphics[scale=0.29]{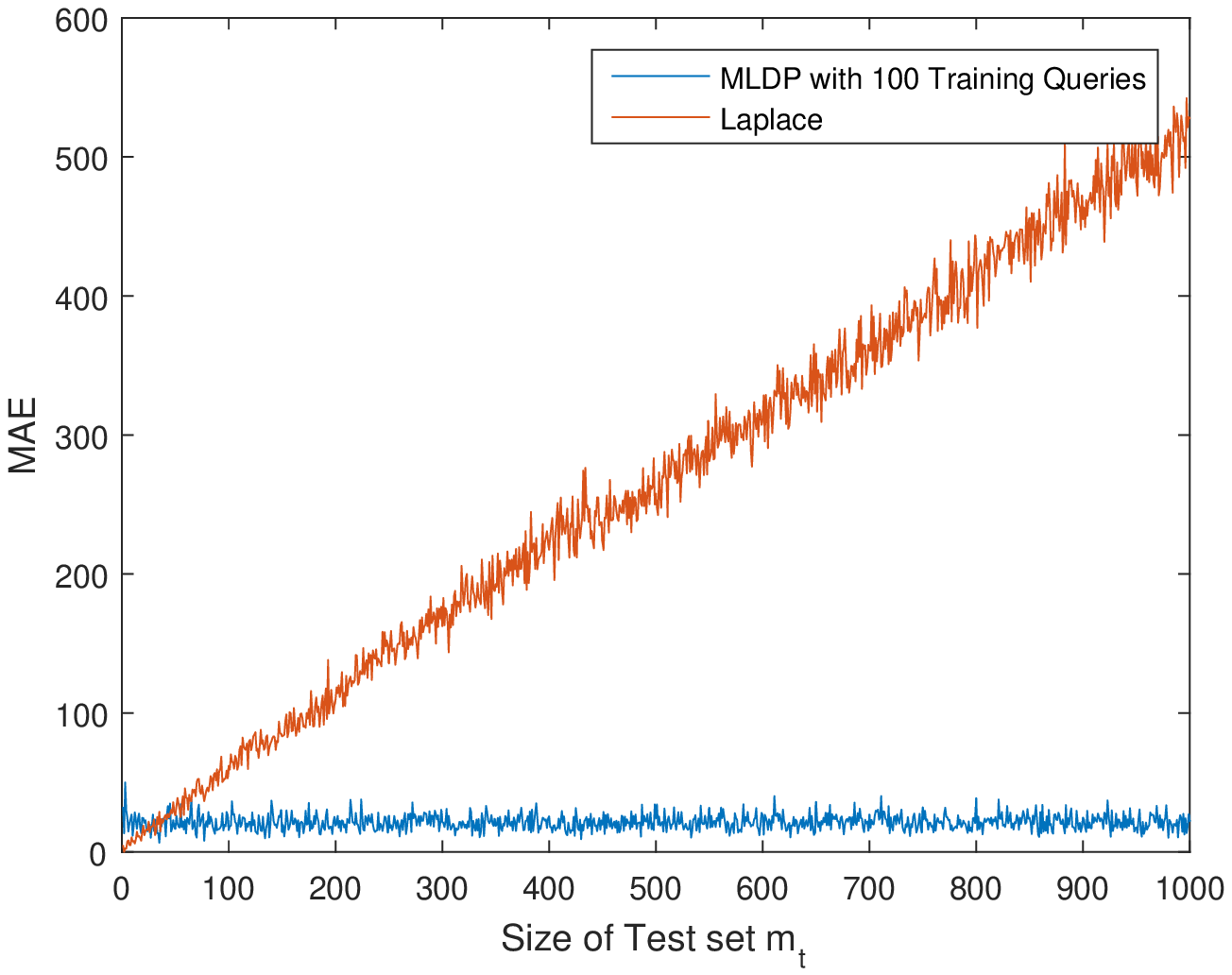}
}
\subfloat[\texttt{Simulated Histogram}]{
\label{Fig-Simu-test-200train}
\includegraphics[scale=0.29]{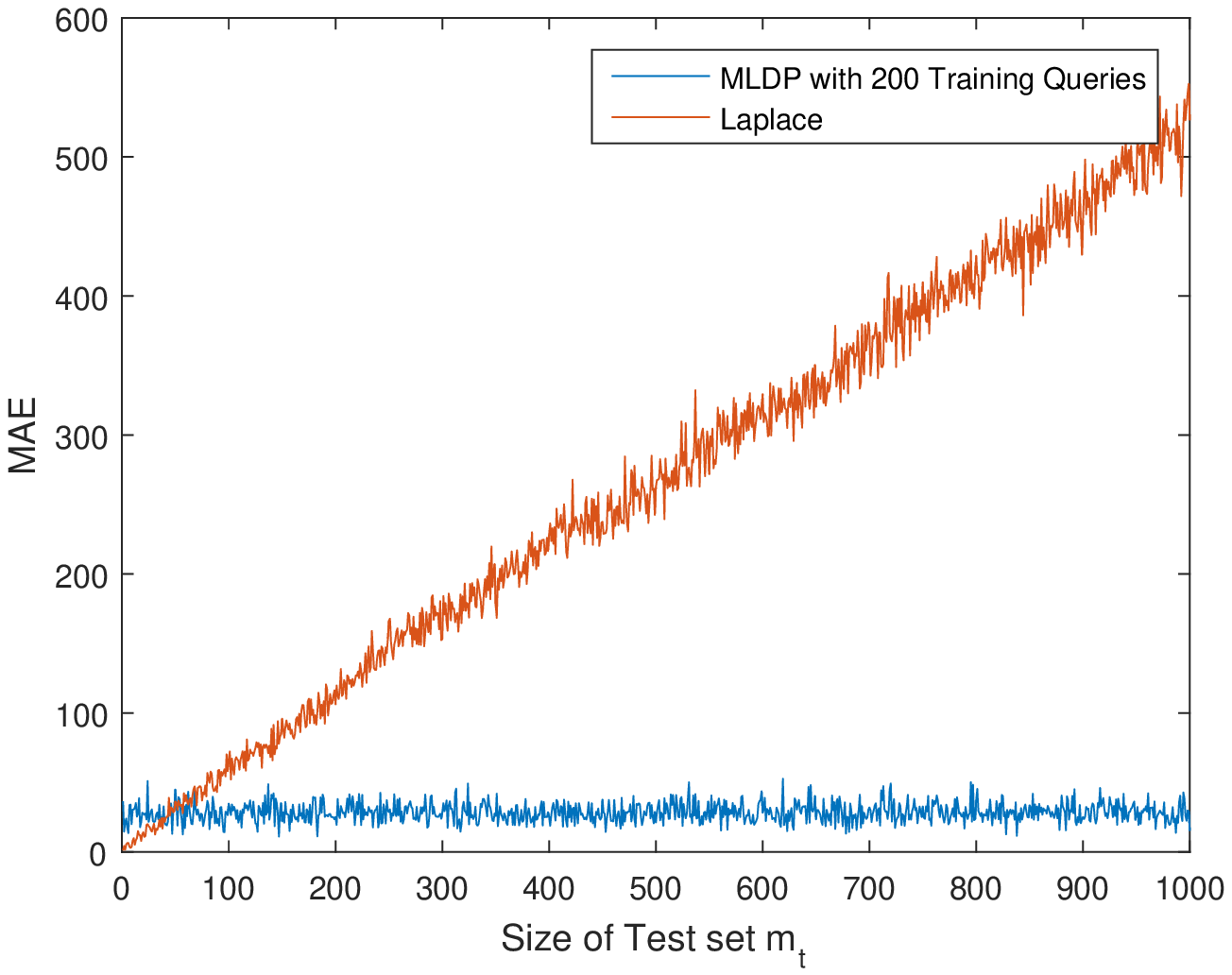}
}
\caption{Performance of \emph{MLDP} in different sizes of test set $\epsilon$}
\label{FIG-Test}
\end{figure*}

The size of the test set indicates the number of queries that a method needs to answer.
This set of experiments examines the performance of $MLDP$ in various sizes of test set,
in which $m_{t}$ is varied from $1$ to $1000$.
As the best size of training set varies with different datasets,
we set $m$ to $100$, $200$ and $300$.
The $MAE$ result is compared with the traditional \emph{Laplace} method and the $\epsilon$ is fixed at $1$.

Fig.~\ref{FIG-Test} shows the impact of the size of test set on
the performance of the \emph{MLDP} and the \emph{Laplace} method.
For all datasets,
the \emph{MAE}s of \emph{MLDP} remain stable with increases in the test set.
However,
the \emph{MAE}s of the \emph{Laplace} method increase linearly
with the enhancement of the test set.
This is because the total sensitivity increases linearly
with the growth in test set size.
When the privacy budget is fixed,
the volume of noise added to the test query answers is raised linearly.

We can also observe in Fig.~\ref{FIG-Test} that
when the size of the test set is small,
the \emph{MAE} of \emph{MLDP} is larger than the \emph{MAE} of the \emph{Laplace} method.
When the size is increased,
\emph{MLDP} outperforms the \emph{Laplace} method.
For example,
Fig.~\ref{Fig-Nettrace-test-100train} shows that when $m_{t}<30$,
\emph{MLDP} has a higher \emph{MAE} than the \emph{Laplace} method.
In Fig.~\ref{Fig-Nettrace-test-200train},
we observe that when the size of the training set is $200$,
meaning that the model is less accurate than that in Fig.~\ref{Fig-Nettrace-test-100train},
\emph{MLDP} has a higher \emph{MAE} than the \emph{Laplace} method until $m_{t}=300$.
This result means that \emph{MLDP} will be more suitable for publishing large set of queries.
In other circumstances,
\emph{MLDP} may demonstrate worse performance than the traditional \emph{Laplace} method.

This trend can also be observed in other datasets.
Fig.~\ref{Fig-Search-test-100train} shows that
\emph{MLDP} performs better on the Search log dataset when
the size of test set reaches $250$.
However,
if the model can be trained with $200$ training queries,
the prediction model will be improved significantly.
Fig.~\ref{Fig-Search-test-100train} shows the result
when the size of the training set is increased to $200$.
\emph{MLDP} outperforms the \emph{Laplace} method when $m>30$.
From Fig.~\ref{Fig-Social-test-200train}, Fig.~\ref{Fig-Social-test-300train},
Fig.~\ref{Fig-Simu-test-100train} and Fig.~\ref{Fig-Simu-test-200train},
we can conclude that when publishing a large set of queries,
\emph{MLDP} is more suitable than the \emph{Laplace} method.
However,
if a single query is being published,
\emph{MLDP} may not outperform the \emph{Laplace} method as
the sensitivity of the single query is relatively lower.
The \emph{Laplace} method will introduce a lower volume of noise,
while \emph{MLDP} may result in more model errors.

\subsection{Performance in Different Learning Algorithms}\label{sub-alg}

Apart from the traditional \emph{Laplace} mechanism,
we also compare \emph{MLDP} with other two prevalent methods.
The first is \emph{Matrix}~\cite{Li2013ICDT}, which aims to decrease the correlation between batches of queries.
The second is \emph{PMW}~\cite{Hardt201061}.
It is one of the most populate iterative publishing
methods in the differential privacy community.
As both algorithms can only deal with count or range queries,
we only test range queries with $\epsilon$ from $0.1$ to $1$.

As shown in Fig.~\ref{FIG-Comp},
we observe that \emph{MLDP} has a lower \emph{MAE} on all values of $\epsilon$ in all datasets.
For example,
in Fig.~\ref{Fig-Nettrace-Comp},
when $\epsilon=0.3$,
\emph{MLDP} has an \emph{MAE} of $26.5437$ while the \emph{Laplace} method has $106.3405$,
with an improvement of $79.7968$.
When $\epsilon = 1$,
\emph{MLDP} achieves an \emph{MAE} of $8.4615$
and outperforms \emph{Laplace} which has $32.2350$.
These results imply that
\emph{MLDP} outperforms traditional \emph{Laplace} when answering a large set of queries.
When compared with other methods,
\emph{MLDP} still has lower \emph{MAEs}.
When $\epsilon=0.3$,
the \emph{PMW} has a \emph{MAE} of $72.0941$ and \emph{Matrix} has a \emph{MAE} of $65.8152$,
which is higher than \emph{MLDP}.
This trend is consistent with the increase in $\epsilon$.
When $\epsilon$ reaches $1$,
\emph{PMW} is $22.3119$ and \emph{Matrix} is $19.3381$,
both of which are higher than that of \emph{MLDP} with \emph{MAE}$=8.4615$.

The improvement achieved by \emph{MLDP} can also be observed in Fig.~\ref{Fig-Search-Comp},~\ref{Fig-Social-Comp}, and Fig.~\ref{Fig-Simu-Comp}.
The proposed \emph{MLDP} mechanism has better performance because
the prediction process for answering test queries does not consume any privacy budget,
while noise is only added in the training queries.
The traditional \emph{Laplace} method consumes the privacy budget
when answering every query in the test set, and the sensitivity
is affected by the correlation between large sets of queries,
which leads to inaccurate answers.
The experimental results show the effectiveness of \emph{MLDP}
in answering a large set of queries.
It is also worth noting that this test set is unknown by \emph{MLDP} in the training process,
but for \emph{PMW}, \emph{Matrix} and \emph{Laplace},
it should be provided before publishing.
This shows that \emph{MLDP} can deal with unknown queries, while other methods cannot.

In the context of \emph{differential privacy},
the privacy budget $\epsilon$ is a key parameter for determining the level of privacy.
From Fig.~\ref{FIG-Comp},
we can also check the impact of $\epsilon$ on the performance of \emph{MLDP}.
According to Dwork~\cite{Dwork1791836},
$\epsilon = 1$ or less would be suitable for privacy preservation purposes,
and we follow this rule in our experiments.
For a comprehensive investigation,
we evaluate \emph{MLDP}'s performance at various privacy preservation levels,
by varying the privacy budget $\epsilon$ from $0.1$ to $1$
with a $0.1$ step on four datasets.
It is observed that
as $\epsilon$ increases,
the \emph{MAE} evaluation becomes better,
which means that
the lower the privacy preservation level,
the better the utility.
In Fig.~\ref{Fig-Search-Comp},
the \emph{MAE} of \emph{MLDP} is $119.2693$ when $\epsilon=0.1$.
Even though it preserves a strict privacy guarantee,
the query answer is inaccurate.
When $\epsilon=0.7$,
the \emph{MAE} drops to $18.9745$,
retaining an acceptable utility in the result.
The same trend can be observed on other datasets.
For example,
when $\epsilon=0.7$,
the \emph{MAE} is $45.3686$ in Fig.~\ref{Fig-Social-Comp},
and is $34.7118$ in Fig.~\ref{Fig-Simu-Comp}.
Both show great improvement compared to $\epsilon=0.1$.
These results confirm that
the utility is enhanced as the privacy budget increases.

We observe that the
\emph{MAE} decreases faster when $\epsilon$ ascends from $0.1$ to $0.4$,
than when $\epsilon$ ascends from $0.4$ to $1$.
This indicates that a larger utility cost is needed
to achieve a higher privacy level ($\epsilon=0.1$).
We also observe that
\emph{MLDP} and other methods perform stably
when $\epsilon \geq 0.7$.
This indicates that
\emph{MLDP} is capable of retaining the utility for data release
while satisfying a suitable privacy preservation requirement.

The evaluation shows the effectiveness of the \emph{MLDP} method from several aspects.
1) It retains a higher accuracy compared to other methods when answering large sets of queries.
2) Its performance is significantly enhanced
with the increase in the privacy budget.
We can select a suitable privacy budget to achieve a better trade-off.
3) With a sufficient privacy budget,
the utility loss can be trivial.

\begin{figure}[htbp]
\centering
\subfloat[\texttt{Nettrace}]{
\label{Fig-Nettrace-Comp}
\includegraphics[scale=0.29]{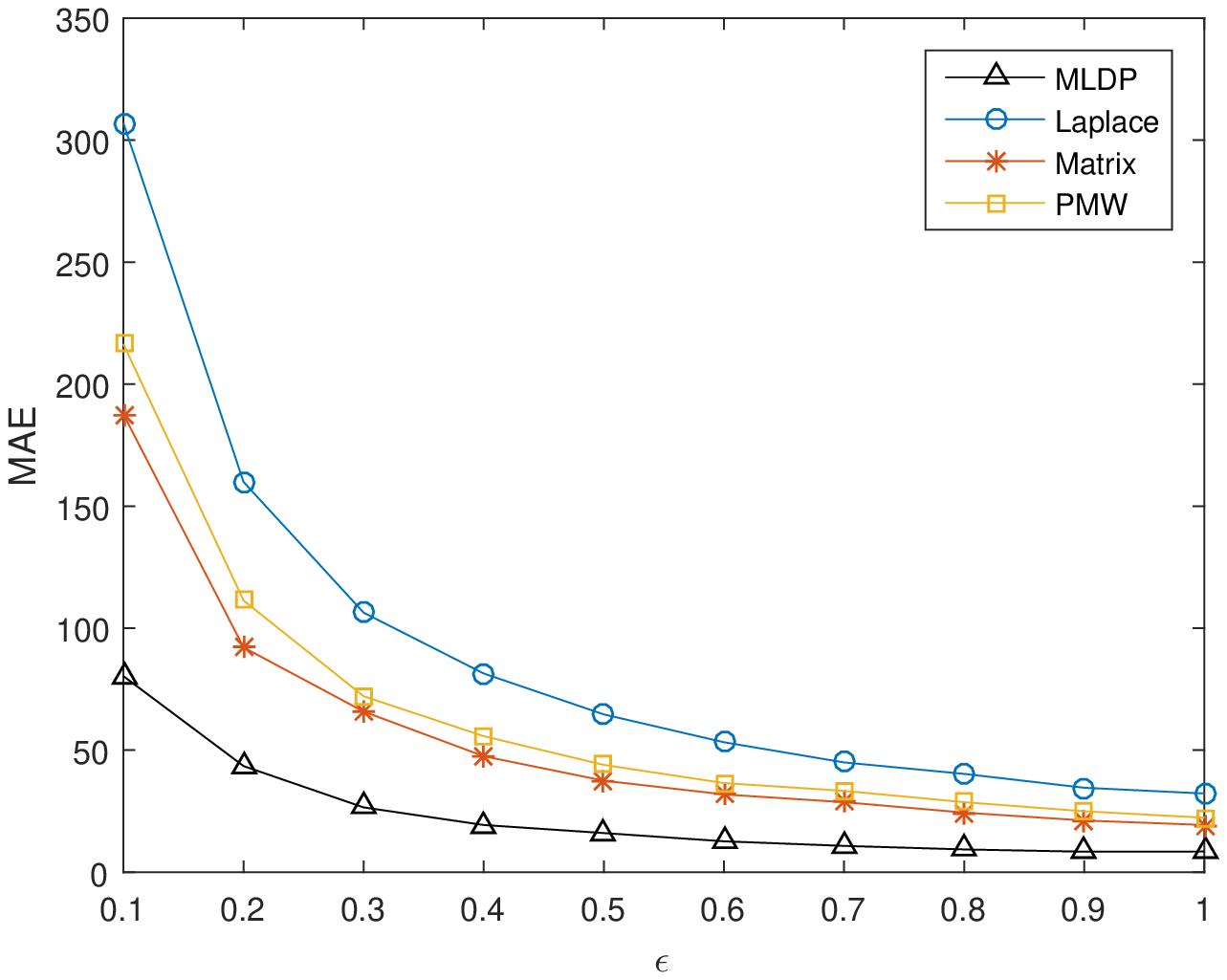}
}
\subfloat[\texttt{Search Log}]{
\label{Fig-Search-Comp}
\includegraphics[scale=0.29]{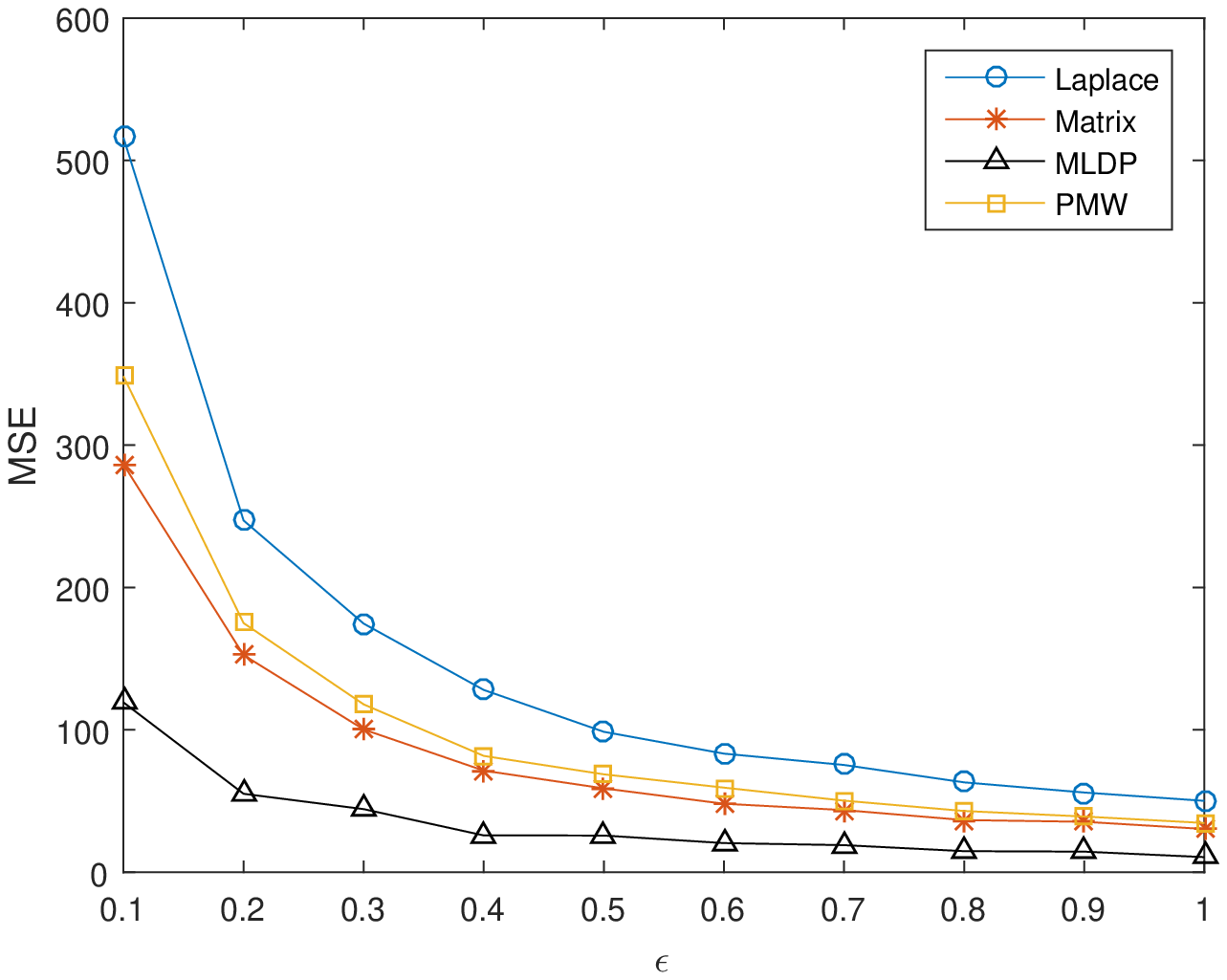}
}

\subfloat[\texttt{Social Network}]{
\label{Fig-Social-Comp}
\includegraphics[scale=0.29]{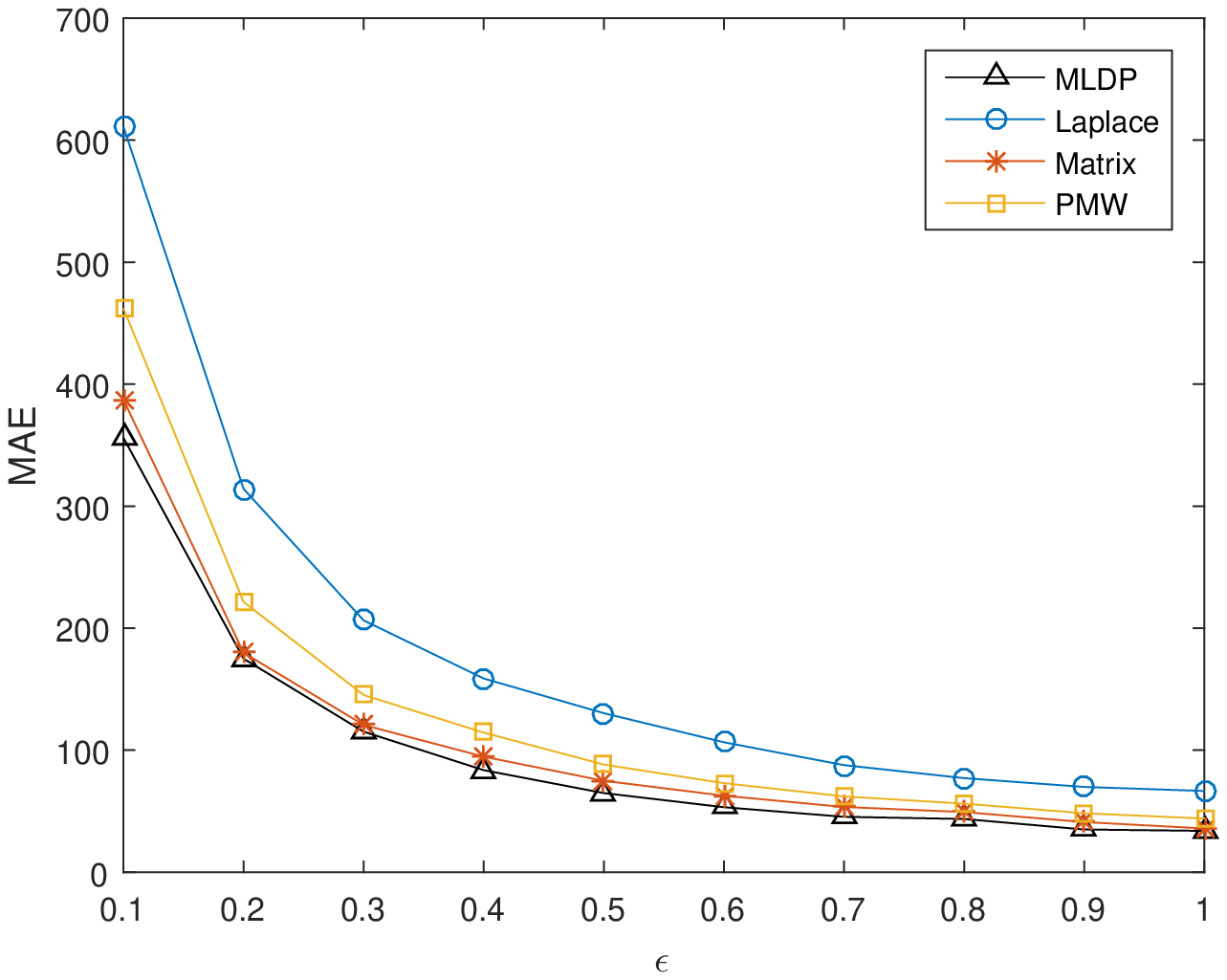}
}
\subfloat[\texttt{Simulated Histogram}]{
\label{Fig-Simu-Comp}
\includegraphics[scale=0.29]{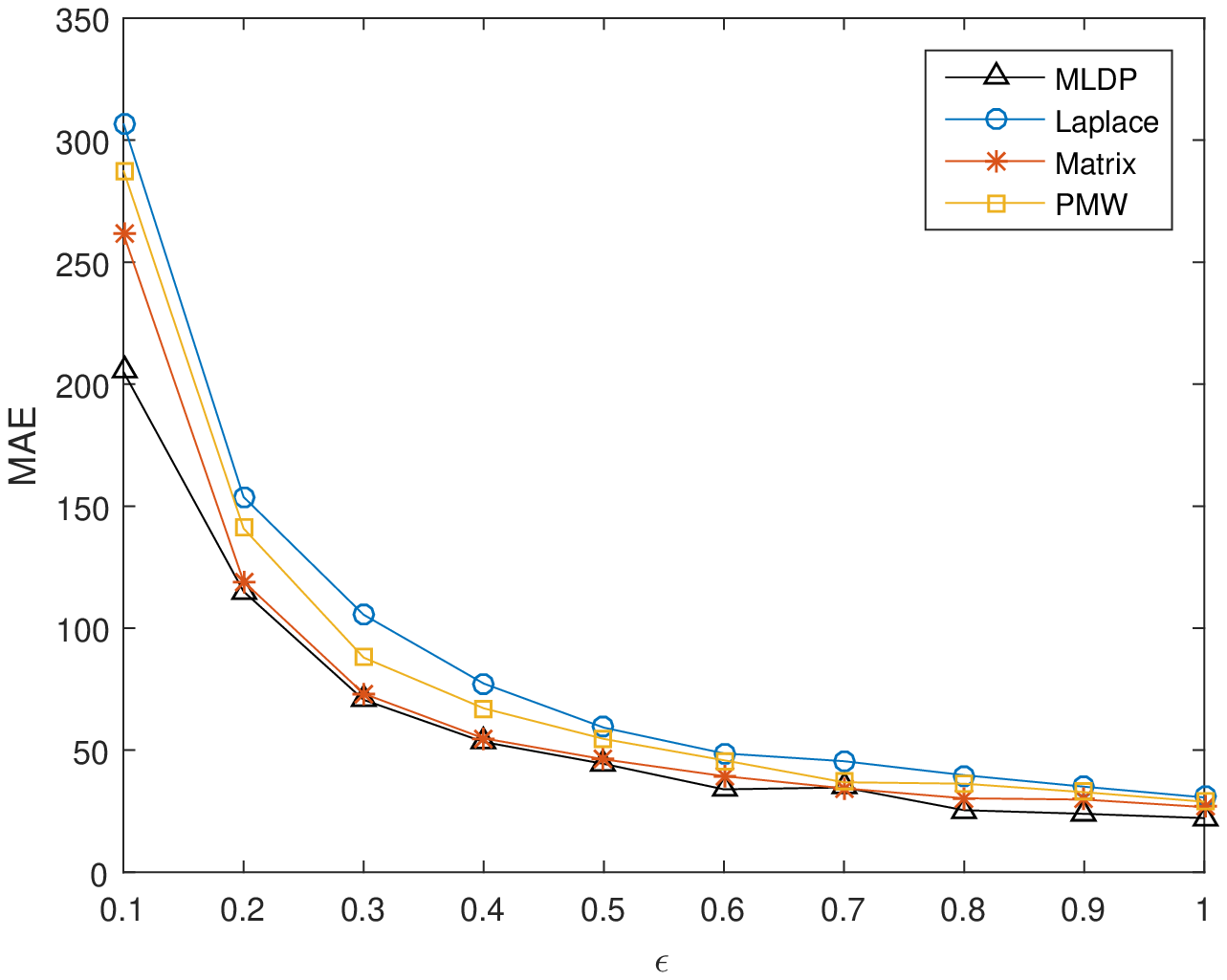}
}
\caption{Method Comparison under different $\epsilon$}
\label{FIG-Comp}
\end{figure}

\section{Related Work}

A plethora of methods has been proposed for differentially private data publishing.
Among them,
two different types of method exist that preserve differential privacy
for \emph{non-interactive} data publishing.
One type of method is synthetic dataset publishing another type is the batch queries publishing.

Synthetic dataset publishing attempts to publish a perturbed dataset instead of the original one.
Mohammed~{et~al.}~\cite{Mohammed2011KDD} proposed an anonymized algorithm \emph{DiffGen}
to preserve privacy for data mining purposes.
The anonymization process satisfies the constraint of differential privacy.
Zhang~{et~al.}~\cite{Zhang14SIGMOD} assumed that there are correlations between attributes.
If these correlations can be modeled,
the model can be used to generate a set of marginals to
simulate the distribution of the original dataset.
Chen~{et~al.}~\cite{Chen15KDD} addressed the similar problem
by proposing a clustering generated method.
%They first learned the pairwise correlation of all attributes
%and generated a dependency graph.
%Secondly,
%they applied the junction tree algorithm to the dependency graph to
%identify a collection of attribute clusters, to derive all the noisy marginals.
%Lastly, they made use of the noisy marginal tables and the inference model
%to generate a synthetic dataset.
All of above works are focusing on the dataset perturbation, there is another line of works
concerning on the sampling method based on the learning theory.

With the learning theory development,
Kasiviswanathan~\cite{Kasiviswanathan2008531} claimed that
almost anything learnable can be learned privately.
Blum~{et~al.}~\cite{Blum2008LTA} subsequently claimed that
the main purpose of analyzing a dataset is to
obtain information about a certain concept.
%If the query on a published dataset is limited to a particular concept,
%the learning process can ensure the accuracy of the query output.
Based on their theories,
Kasiviswanathan~{et~al.}~\cite{Kasiviswanathan2008531}
proposed an \emph{Exponential}-based mechanism to search a synthetic dataset from the data universe
that is able to accurately answer a group of queries.
Blum~{et~al.}~\cite{Blum2008LTA} applied a similar \emph{Exponential}-based mechanism,
\emph{Net} mechanism,
to generate a synthetic dataset over a discrete domain.

%The synthetic dataset was considered to be a difficult problem for a long time
%due to the large amount of introduced noise, which
%leads to inaccurate output.
%From this perspective,
%publishing a synthetic dataset using \emph{Laplace} mechanism is difficult.

Another type of method is to release a batch of queries instead of a dataset.
The traditional \emph{Laplace} method belongs to this type,
but it introduces a large amount of noise due to the correlation between queries.
Current research works focus on how to decrease the correlation between batches of queries,
so that the total sensitivity can be diminished.

Xiao~{et~al.}~\cite{Xiao2011} proposed a wavelet transformation,
called \emph{Privelet},
on the dataset to decrease the sensitivity.
Li~{et~al.}~\cite{Li2013ICDT} proposed the \emph{Matrix} mechanism which
answer sets of linear counting queries.
%The sets of queries are represented by a matrix $A$ called \emph{workload},
%in which each row contains the coefficients of a linear query.
Given a set of queries,
the \emph{Matrix} mechanism defines a workload $A$ accordingly
and obtains noisy answers by implementing the \emph{Laplace} mechanism.
The estimates are then used on the $A$
to generate estimates of the submitted queries.
Huang~{et~al.}~\cite{Huang2015DB} transformed the query sets to
a set of orthogonal queries to reduce the correlation between queries.
The correlation reduction helps to decrease the sensitivity of the query set.
Yuan~{et~al.}~\cite{Yuan2015batch} presented a low-rank mechanism (LRM),
an optimization framework that minimizes the overall error of the results
for a batch of linear queries.

The method in this paper is neither similar to synthetic dataset publishing, nor to batch query publishing.
The proposed method aims to publish a model rather than a synthetic dataset or query answers.
Unlike previous work,
our work aims to publish a model to answer fresh queries.
It is entirely new thinking on the data publishing problem
which transfers data publishing into a machine learning process.
Many challenges in non-interactive data publishing can thus be overcome by
using existing flourishing machine learning theories.

\section{Conclusions} \label{sec-conclusions}

Differential privacy is an influential notion in the research of privacy preserving data publishing,
but the
existing differentially private method fails to provide
accurate results for publishing large numbers of queries.
Two challenges must be tackled in this process:
how to decrease the correlation between queries and how to deal with
unknown queries before publishing.
This paper proposes a query learning solution to deal
with both challenges and makes the
following contributions:
We propose a novel \emph{MLDP} method to transfer the
data publishing problem to a machine learning problem and
prove the accuracy bound of the \emph{MLDP}.
The \emph{MLDP} method exploits a possible way to publish data structures.
Extensive experiments has been used on both real and synthetic datasets
to prove the effectiveness of the proposed \emph{MLDP}.
These contributions not only form a practical solution for
non-interactive data publishing in terms of higher accuracy,
but also propose a possible way to release various types of data in the future.

% ----------------------------------------------------------------
% ----------------------------------------------------------------
%\newpage
%\bibliographystyle{plain}
%\bibliography{references}
%=================================================================

\end{document}